\documentclass[12pt]{article}

\newif\ifpreprint   \preprinttrue
\newif\ifbiblatex   \biblatexfalse

\usepackage{amsthm,amsmath,amssymb,mathtools}
\usepackage{graphicx,booktabs,subcaption}
\usepackage[dvipsnames]{xcolor}
\usepackage{microtype}

\ifpreprint
  \usepackage{setspace}
  \usepackage[margin=1in]{geometry}
  \usepackage{authblk}
  \usepackage{newtxtext,newtxmath}   % Times text + math
\else
  \providecommand{\affil}[2][]{}     % no-op when authblk absent
\fi

\ifbiblatex
\else
  \usepackage{natbib}
\fi

\usepackage{my-macros}
\newcommand{\indep}{\perp\!\!\!\perp}
\newcommand{\nindep}{\not\!\perp\!\!\!\perp}

\ifpreprint
  \usepackage[colorlinks=true,linkcolor=Maroon,
    urlcolor=MidnightBlue,citecolor=MidnightBlue]{hyperref}
\else
  \usepackage[colorlinks=false]{hyperref}
\fi
\usepackage[capitalise,noabbrev]{cleveref}

\providecommand{\keywords}[1]{\small\textbf{\textit{Keywords---}} #1}

\title{Nonparametric Point Identification of Treatment Effect Distributions via Rank Stickiness}

\author{Tengyuan Liang}
\affil{The University of Chicago}
\date{\today}

\begin{document}
\maketitle

\begin{abstract}
Treatment effect distributions are not identified without restrictions on the joint distribution of potential outcomes. Existing approaches either impose rank preservation---a strong assumption---or derive partial identification bounds that are often wide. We show that a single scalar parameter, rank stickiness, suffices for nonparametric point identification while permitting rank violations. The identified joint distribution---the coupling that maximizes average rank correlation subject to a relative entropy constraint, which we call the Bregman-Sinkhorn copula---is uniquely determined by the marginals and rank stickiness. Its conditional distribution is an exponential tilt of the marginal with a Bregman divergence as the exponent, yielding closed-form conditional moments and rank violation probabilities; the copula nests the comonotonic and Gaussian copulas as special cases. The empirical Bregman-Sinkhorn copula converges at the parametric $\sqrt{n}$-rate with a Gaussian process limit, despite the infinite-dimensional parameter space. We apply the framework to estimate the full treatment effect distribution, derive a variance estimator for the average treatment effect tighter than the Fr\'{e}chet--Hoeffding and Neyman bounds, and extend to observational studies under unconfoundedness.
\end{abstract}

\keywords{treatment effect distribution, identification, heterogeneity, imputation, nonparametric copula, sensitivity analysis.}

%%%%%%%%%%%%%%%%%%%%%%%%%%%%%%%%%%%%%%%%%%%%%%
%% Main Paper                               %%
%%%%%%%%%%%%%%%%%%%%%%%%%%%%%%%%%%%%%%%%%%%%%%
% \tableofcontents

\section{Introduction}

When the goal of a policy evaluation is to understand who benefits and who is harmed, the Average Treatment Effect (ATE) or Quantile Treatment Effects (QTE) are insufficient---what is needed is the entire Treatment Effect Distribution (TED). Consider a job training program: a positive ATE on earnings may obscure the fact that gains are concentrated among high-skilled workers while low-skilled workers experience adverse effects. The TED thus reveals the complete pattern of individual heterogeneity in program impacts.

However, the TED is not identified without additional assumptions: only one potential outcome is observed per unit, so the joint distribution $(Y_0, Y_1) \sim \pi$ is never observed. The existing literature faces a sharp trade-off between the strength of assumptions and informativeness of conclusions.
At one extreme, \citet{manski1990nonparametric,manski1997monotone} derive partial identification bounds under minimal assumptions, but these worst-case bounds are often wide \citep{blundell2007changes}. Sharper bounds exploit the coupling structure: \citet{HeckmanSmithClements1997} and \citet{fan2010sharp} employ Fr\'{e}chet--Hoeffding copula bounds; \citet{kim2014identifying} tighten these under shape restrictions; \citet{firpo2019partial} derive bounds on functionals of the joint distribution.
At the other extreme, rank preservation---the requirement that each individual's rank in the outcome distribution is unchanged by treatment,
$$
(Y_0 - Y_0') \cdot (Y_1 - Y_1') > 0 \quad \text{a.s.\ for independent draws } (Y_0, Y_1), (Y_0', Y_1') \sim \pi \;,
$$
---yields point identification: the TED reduces to the quantile-by-quantile difference of the two marginals \citep{firpo2007efficient}; conditional variants have been studied by \citet{chernozhukov2005iv} and \citet{athey2006identification}. Yet \citet{HeckmanSmithClements1997} argue this assumption is implausible in many empirical settings; indeed, under existing approaches, even mild departures from rank preservation undermine nonparametric point identification of the TED.

This paper shows that point identification does not require rank preservation. We introduce a single structural parameter---rank stickiness $\rho$---that governs the degree of departure from comonotonicity and suffices for nonparametric point identification of the TED given the marginals. Rank stickiness captures the empirical regularity that treatment shifts the distribution of outcomes while individual ranks exhibit persistence, a phenomenon well-documented in earnings mobility both within individual lifetimes \citep{kopczuk2010earnings} and across generations \citep{chetty2014united}. We treat $\rho$ as a sensitivity parameter: rather than assuming a single value, we characterize how the identified TED varies with $\rho$.

Formally, we model the joint distribution of $(Y_0, Y_1) \sim \pi$ as the coupling that maximizes average rank correlation subject to a relative entropy constraint parameterized by $\rho \in (0,1]$:
$$
\max_{\pi \in \Pi(\mu, \nu)~ \text{s.t.}~ \mathsf{KL}(\pi \mid \mu \otimes \nu) \leq R(\rho)} ~\E_{(Y_0, Y_1) \sim \pi} \E_{(Y_0', Y_1') \sim \pi} (Y_0 - Y_0') \cdot(Y_1 - Y_1')\;,
$$
where $\Pi(\mu, \nu)$ is the set of couplings with marginals $\mu, \nu$, and $R(\rho)$ is a monotone function of $\rho$ that calibrates the degree of rank stickiness. We call the solution the Bregman-Sinkhorn copula. The conditional distribution of $Y_1 \mid Y_0$ takes the form of an exponential tilting of the treatment marginal, with exponent given by a Bregman divergence---yielding closed-form conditional moments and rank violation probabilities. The copula nests the comonotonic copula as $\rho \to 1$ and generalizes the Gaussian copula to the nonparametric setting.

The Bregman-Sinkhorn copula is nonparametrically identified by $(\mu, \nu, \rho)$ alone under mild regularity conditions on the marginals. Although rank violations occur with non-vanishing probability at every $\rho \in (0,1)$, the conditional rank of $Y_1$ given $Y_0$ remains sticky around the identity as $\rho \to 1$, with an explicit Bregman-divergence characterization of the conditional distribution. The relaxation of rank preservation incurs no loss in statistical efficiency: the empirical Bregman-Sinkhorn copula converges at the parametric $\sqrt{n}$-rate with a Gaussian process limit for any $\rho \in (0,1]$, despite the infinite-dimensional parameter space.

Applied to causal inference, efficient estimation of $\pi_\rho$ yields the entire treatment effect distribution $\cL(Y_1 - Y_0)$ and any functional of the joint distribution $(Y_0, Y_1)$. We derive a closed-form variance estimator for the ATE that is tighter than the Fr\'{e}chet--Hoeffding and Neyman bounds, and verify this in simulations that also examine sensitivity of the estimated TED to $\rho$. Empirically, we find the comonotonic copula ($\rho = 1$) simultaneously underestimates treatment effect heterogeneity and overestimates the variance of the ATE estimator; the Bregman-Sinkhorn copula corrects both. The framework extends to observational studies: under unconfoundedness, a covariate-specific stickiness parameter identifies the conditional joint distribution stratum by stratum, and the per-stratum transport map is computed via a gradient flow whose driving term can be estimated via binary classification on the observed data.

At a broader level, access to the joint distribution of potential outcomes goes well beyond what conditional average treatment effects can deliver \citep{athey2017state,athey2019machine}: the conditional mean identifies which subgroups benefit on average \citep{FarrellLiangMisra2021,athey2021policy}, but the TED reveals the probability that a given individual is harmed by treatment---information essential for targeting and policy design.

We conclude this section by collecting notation used throughout. Let $\cY$ be a compact subset of $\R^d$, $\sP(\cY)$ the set of Borel probability measures on $\cY$, and $\Pi(\mu, \nu)$ the set of couplings with marginals $\mu, \nu \in \sP(\cY)$. We denote by $\cL(Y)$ the law of a random variable $Y$ and by $T \sharp \mu$ the pushforward of $\mu$ under a measurable map $T$, defined by $(T \sharp \mu)(B) = \mu(T^{-1}(B))$ for every Borel set $B$. For scalar-valued potential outcomes $(Y_0, Y_1) \sim \pi \in \Pi(\mu, \nu)$, we write $F(y_0) = \mu((-\infty, y_0])$ and $G(y_1) = \nu((-\infty, y_1])$ for the marginal distribution functions, and $F^{-1}$, $G^{-1}$ for the corresponding quantile functions. The stickiness parameter is $\rho \in (0, 1]$, and we set $\epsilon(\rho) := (1-\rho)/\rho$ unless otherwise stated.

\section{Rank Violations and Optimal Transport}
\label{sec:rank-violations}

For a given individual, let $Y_0$ and $Y_1$ denote the potential outcomes under control and treatment, respectively; the individual treatment effect is $Y_1 - Y_0$. Since $Y_0$ and $Y_1$ are never jointly observed for the same unit, the law of $Y_1 - Y_0$ is not identified without additional restrictions on the coupling $(Y_0, Y_1) \sim \pi \in \cP(\cY \times \cY)$.

A canonical identifying restriction for scalar outcomes is \emph{strict rank preservation}:
$$
  (Y_0, Y_1) \sim \pi = (F^{-1}, G^{-1})\sharp U, \qquad U \sim \mathrm{Unif}(0,1) \;.
$$
This is a strong restriction: it implies $Y_1 = G^{-1} \circ F(Y_0)$ almost surely. Equivalently, if a unit is at quantile $u$ of the control-outcome distribution, that unit is also at quantile $u$ of the treatment-outcome distribution. \citet{HeckmanSmithClements1997} argue that this assumption is often implausible. We therefore consider a weaker notion, which we call \emph{average rank preservation}.

\subsection{Optimal Transport and Average Rank Preservation}
We define average rank preservation in the general multivariate setting ($\cY = \R^d$).

\begin{definition}[Average rank preservation]
A coupling $\pi \in \Pi(\mu, \nu)$ satisfies average rank preservation if
\begin{align*}
  \sR(\pi) := \frac{1}{2} \E_{(Y_0, Y_1) \sim \pi} \E_{(Y'_0, Y'_1) \sim \pi} \langle Y_0 - Y'_0 , Y_1 - Y'_1 \rangle > 0 \;,
\end{align*}
where $(Y_0, Y_1)$ and $(Y'_0, Y'_1)$ are independent draws from $\pi$. This quantity is the expected inner product between the control-outcome gap and the treatment-outcome gap for two independent units.
\end{definition}

The following notion ties the average rank preservation measure to the Kullback-Leibler divergence of $\pi$ from the independence coupling.
\begin{definition}[$\rho$-rank preservation]
For $\rho \in (0, 1]$ and $\epsilon(\rho) := (1-\rho)/\rho$, a coupling $\pi \in \Pi(\mu, \nu)$ satisfies $\rho$-rank preservation if
\begin{align*}
   \sR(\pi) >  \epsilon(\rho) \, \mathsf{KL}\bigl( \pi \mid \mu \otimes \nu \bigr) \;.
\end{align*}
$\rho$-rank preservation generalizes average rank preservation: the case $\rho = 1$ reduces to the condition $\sR(\pi) > 0$.
\end{definition}

We now formalize the connection between rank-preservation and optimal transport (OT).
\begin{proposition}[Rank correlation and entropic OT]
  \label{prop:maximize-average-rank}
Suppose $\mu, \nu \in \sP(\cY)$ with bounded second moments.
\begin{enumerate}
  \item[(i)] Suppose $\mu$ is absolutely continuous with respect to the Lebesgue measure. Then the unique maximizer of $\sR(\pi)$ over $\Pi(\mu, \nu)$ is the Monge--Kantorovich coupling $\pi_1 = (\mathop{id}, \nabla\phi_1)\sharp\mu$, where $\nabla\phi_1$ is the Brenier map from $\mu$ to $\nu$.
  \item[(ii)] For each $\rho \in (0,1)$ and $\epsilon(\rho) := (1-\rho)/\rho$, maximizing $\sR(\pi) - \epsilon(\rho) \,\mathsf{KL}(\pi\mid \mu\otimes\nu)$ over $\Pi(\mu,\nu)$ is equivalent to solving the entropic OT problem
  $$
    \min_{\pi \in \Pi(\mu, \nu)} \int_{\cY \times \cY} \frac{1}{2}\| y_0 - y_1 \|^2 \dd \pi(y_0, y_1)  + \epsilon(\rho) \, \mathsf{KL}(\pi \mid \mu \otimes \nu) \;.
  $$
\end{enumerate}
\end{proposition}

\subsection{Rank Violations}
Consider the scalar case with $\cY = \R$. The coupling maximizing $\rho$-rank preservation exhibits rank violations for every $\rho \in (0, 1)$.

For the case $\rho = 1$, maximizing average rank preservation $\sR(\pi)$ over $\Pi(\mu, \nu)$ uniquely identifies the Monge--Kantorovich coupling $\pi_1 = (\mathop{id}, G^{-1} \circ F)\sharp\mu$. This coupling exhibits perfect rank preservation: if a unit is at quantile $u$ of the control-outcome distribution with $Y_0 = F^{-1}(u)$, that unit is also at quantile $u$ of the treatment-outcome distribution with $Y_1 = G^{-1}(u)$.

For the case $\rho \in (0, 1)$, we will prove in Theorem~\ref{thm:identifiability} that under mild regularity conditions on $\mu, \nu$, maximizing $\rho$-rank preservation over $\Pi(\mu, \nu)$ uniquely identifies a nonparametric coupling $\pi_\rho$
\begin{align}
\label{eqn:pi-rho}
\pi_\rho := \argmax_{\pi \in \Pi(\mu, \nu)} \sR(\pi) - \epsilon(\rho) \,\mathsf{KL}(\pi\mid\mu\otimes\nu) \;.
\end{align}
\begin{proposition}[Rank violations]
  \label{prop:rank-violations}
  Let $\cY \in \R$ be a bounded interval, and let $\mu$ and $\nu$ be absolutely continuous with respect to the Lebesgue measure with strictly positive densities.
  For $\rho \in (0, 1)$, consider the coupling $(Y_0, Y_1) \sim \pi_\rho \in \Pi(\mu, \nu)$ uniquely defined in \cref{eqn:pi-rho}.
  Then for any $u \in (0, 1)$,
  $$
  \PP\Bigl( Y_1 \neq G^{-1}(u) \mid Y_0 = F^{-1}(u) \Bigr) > 0  \;.
  $$
\end{proposition}

\section{Bregman-Sinkhorn Copula}
\label{sec:bregman-sinkhorn-copula}

Before formally deriving the nonparametric identification result under regularity conditions, we introduce the key properties of the Bregman-Sinkhorn copula, taking uniqueness of the solution to \cref{eqn:pi-rho} as given; the proof of uniqueness is deferred to Theorem~\ref{thm:identifiability}.

\subsection{Sinkhorn Copula}
Let $\epsilon(\rho) := (1-\rho)/\rho \in (0, \infty)$. The Sinkhorn dual potentials $\phi_\rho, \psi_\rho : \cY \to \R$ satisfy the system of equations
\begin{align}
  \label{eqn:sinkhorn-system}
\phi_\rho(y_0) &= \epsilon(\rho) \log \int_{\cY} \exp\Bigl( \frac{\langle y_0, y_1 \rangle - \psi_{\rho}(y_1)}{\epsilon(\rho)} \Bigr) \dd \nu(y_1) \;, \nonumber \\
\psi_\rho(y_1) &= \epsilon(\rho) \log \int_{\cY} \exp\Bigl( \frac{\langle y_1, y_0 \rangle - \phi_\rho(y_0)}{\epsilon(\rho)} \Bigr) \dd \mu(y_0) \;.
\end{align}
Under regularity conditions, the pair $(\phi_\rho, \psi_\rho)$ is uniquely identified by the marginals $\mu, \nu$ and the stickiness parameter $\rho$ (to be shown in Theorem~\ref{thm:identifiability}), up to an additive constant $(\phi_\rho,\psi_\rho) \mapsto (\phi_\rho-c, \psi_\rho+c)$.

With these convex dual functions, we define the \emph{Sinkhorn copula} $\pi_{\rho} \in \Pi(\mu, \nu)$ as
\begin{align}
  \label{eqn:sinkhorn-copula}
\frac{\dd \pi_\rho(y_0, y_1)}{\dd \mu(y_0) \dd \nu(y_1)} := \exp\Bigl( \frac{\langle y_0, y_1 \rangle - \phi_{\rho}(y_0) - \psi_{\rho}(y_1)}{\epsilon(\rho)} \Bigr)\;.
\end{align}

The name \emph{Sinkhorn copula} comes from the fact that in the scalar-valued case,
\begin{align*}
\PP(Y_0 \leq y_0, Y_1 \leq y_1)
&= \int_0^{F(y_0)} \dd u_0 \int_0^{G(y_1)} \dd u_1 \, \exp\Bigl( \frac{ F^{-1}(u_0) G^{-1}(u_1) - \phi_{\rho} \circ F^{-1}(u_0) - \psi_{\rho} \circ G^{-1} (u_1)}{\epsilon(\rho)} \Bigr)\;, 
\end{align*}
where the density of the copula satisfies
\begin{align}
  \label{eqn:sinkhorn-copula-density}
c_{\mathrm{S}}^{\rho}(u_0, u_1) := \exp\Bigl( \frac{ F^{-1}(u_0) G^{-1}(u_1) - \phi_{\rho} \circ F^{-1}(u_0) - \psi_{\rho} \circ G^{-1} (u_1)}{\epsilon(\rho)} \Bigr) \;.
\end{align}

\subsection{Bregman Divergence and Conditional Exponential Family}
The name Bregman-Sinkhorn copula reflects the following structure: under $\pi_\rho$, the conditional distribution of $Y_1$ given $Y_0$ is an exponential tilt of the marginal $\nu = \cL(Y_1)$, with a Bregman divergence as the exponent.

\begin{proposition}[Bregman divergence representation]
  \label{prop:bregman-divergence-representation}
  Consider $\cY = \R^d$.
  Let $(Y_0, Y_1) \sim \pi_\rho$ be the Sinkhorn copula defined in \cref{eqn:sinkhorn-copula}, then for any $y_0 \in \cY$, the conditional distribution of $Y_1 \mid Y_0 = y_0$ is given by
  \begin{align}
  \PP( \dd y_1 \mid Y_0 = y_0 )=  \frac{\exp\Bigl( \frac{ - B_{\psi_\rho}(y_1 \mid \nabla \psi^\ast_\rho(y_0))}{\epsilon(\rho)} \Bigr) \dd \nu(y_1) }{ \int_{\cY} \exp\Bigl( \frac{ - B_{\psi_\rho}(y'_1 \mid \nabla \psi^\ast_\rho(y_0))}{\epsilon(\rho)} \Bigr) \dd \nu(y_1') }  \;,
  \end{align}
  where $B_{\psi_\rho}(y_1 \mid y'_1) := \psi_\rho(y_1) - \psi_\rho(y'_1) - \langle y_1 - y'_1, \nabla\psi_\rho(y'_1) \rangle$ is the Bregman divergence associated with $\psi_\rho$, and $\psi_\rho^\ast (y_0) := \sup_{y_1 \in \cY} \{ \langle y_0, y_1 \rangle - \psi_\rho(y_1) \}$ is the convex conjugate of $\psi_\rho$.
\end{proposition}
\begin{proof}[Proof of Proposition~\ref{prop:bregman-divergence-representation}]
By definition of $\pi_\rho$ in \cref{eqn:sinkhorn-copula}
$$
\PP( \dd y_1 \mid Y_0 = y_0 ) := \exp\Bigl( \frac{\langle y_0, y_1 \rangle - \phi_{\rho}(y_0) - \psi_{\rho}(y_1)}{\epsilon(\rho)} \Bigr) \dd \nu(y_1) \;.
$$
By the Sinkhorn equation $\phi_\rho(y_0) = \epsilon(\rho) \log \int_{\cY} \exp\Bigl( \frac{\langle y_0, y_1 \rangle - \psi_{\rho}(y_1)}{\epsilon(\rho)} \Bigr) \dd \nu(y_1)$, we have
\begin{align*}
  \PP( \dd y_1 \mid Y_0 = y_0 ) = \frac{\exp\Bigl( \frac{\langle y_0, y_1 \rangle - \psi_{\rho}(y_1)}{\epsilon(\rho)} \Bigr) \dd \nu(y_1) }{ \int_{\cY} \exp\Bigl( \frac{\langle y_0, y'_1 \rangle - \psi_{\rho}(y'_1)}{\epsilon(\rho)} \Bigr) \dd \nu(y_1') } =  \frac{\exp\Bigl( \frac{ - B_{\psi_\rho}(y_1 \mid \nabla \psi^\ast_\rho(y_0))}{\epsilon(\rho)} \Bigr) \dd \nu(y_1) }{ \int_{\cY} \exp\Bigl( \frac{ - B_{\psi_\rho}(y'_1 \mid \nabla \psi^\ast_\rho(y_0))}{\epsilon(\rho)} \Bigr) \dd \nu(y_1') }  \;.
\end{align*}
Here the last step follows from the fact that $\nabla \psi_\rho^\ast(y_0) = \argmax_{y_1} \{ \langle y_0, y_1 \rangle - \psi_\rho(y_1) \}$ and $\psi^\ast_\rho(y_0) = \langle y_0, \nabla \psi_\rho^\ast(y_0) \rangle - \psi_\rho(\nabla \psi_\rho^\ast(y_0))$,
and thus $$\langle y_0, y_1 \rangle - \psi_\rho(y_1) - \psi^\ast_{\rho}(y_0) = - B_{\psi_\rho}(y_1 \mid \nabla \psi^\ast_\rho(y_0)) \;.$$
\end{proof}

The conditional distribution of $Y_1 \mid Y_0 = y_0$ under the Sinkhorn copula is a member of the natural exponential family. Equivalently, letting $g(y_1) := \dd \nu(y_1) / \dd y_1$ denote the density of $\nu$, the conditional density of $Y_1 \mid Y_0 = y_0$ is
\begin{align}
  \label{eqn:conditional-distribution}
  p(Y_1 = y_1 \mid Y_0 = y_0) &= \exp\Bigl( \frac{\langle y_1, y_0 \rangle - \phi_\rho(y_0) - \psi_\rho(y_1)}{\epsilon(\rho)} \Bigr) g(y_1) \;.
\end{align}
This is a member of the natural exponential family with natural parameter $y_0 / \epsilon(\rho)$ and log-partition function $\phi_\rho(y_0) / \epsilon(\rho)$.
Consequently, the conditional mean and variance of $Y_1 \mid Y_0 = y_0$ are given by the first and second derivatives of $\phi_\rho$.

\begin{proposition}[Conditional moments]
  \label{prop:conditional-moments}
  Let $(Y_0, Y_1) \sim \pi_\rho$, the Bregman-Sinkhorn copula, and let $\phi_\rho$ be a solution to the Sinkhorn equation \eqref{eqn:sinkhorn-system}, then, for any $y_0 \in \cY$,
  \begin{align*}
    \E[Y_1 \mid Y_0 = y_0] &= \nabla \phi_\rho(y_0) \;, \\
    \mathop{Cov}[Y_1 \mid Y_0 = y_0 ] &= \epsilon(\rho) \; \nabla^2 \phi_\rho(y_0) \;.
  \end{align*}
\end{proposition}

\subsection{Connection to Gaussian and Fréchet-Hoeffding Copulas}

We restrict to the scalar case $\cY = \R$ and establish connections to the Fr\'{e}chet--Hoeffding and Gaussian copulas.

Consider the conditional Bregman-Sinkhorn copula: for any $\rho \in (0, 1]$ and $\epsilon \in [0, 1)$, the conditional density of $G(Y_1) = u_1 \mid F(Y_0) = u_0$ is given by
\begin{align}
  \label{eqn:conditional-bregman-sinkhorn-copula}
c_{\mathrm{B}}^{\rho, \epsilon}(u_1 \mid U_0 = u_0) \propto \exp\Bigl( \frac{ - B_{\psi_\rho} \bigl( G^{-1}(u_1) \mid \dot\psi^\ast_{\rho} \circ F^{-1}(u_0)  \bigr) }{\epsilon} \Bigr) \;.
\end{align}
Here $\psi_\rho$ is the solution to the Sinkhorn system, and $\psi_\rho^\ast$ is its convex conjugate.

The Bregman-Sinkhorn copula nests two classical copulas as special cases:

(i) The Bregman-Sinkhorn copula generalizes the Gaussian copula to the nonparametric setting. The Bregman divergence reduces to a Mahalanobis distance when $\psi_\rho(x) = \rho x^2/2$, a quadratic function. With this parametric choice of $(\psi_{\rho}, \psi^\ast_\rho)$, $G = F = \Phi$ the standard Gaussian CDF, and the choice $\epsilon = (1-\rho^2)/\rho$, \cref{eqn:conditional-bregman-sinkhorn-copula} reduces to the conditional distribution of the Gaussian copula with correlation $\rho \in [0, 1)$
\begin{align*}
c_{\mathrm{G}}^{\rho}(u_1 \mid U_0 = u_0) \propto \exp\Bigl( \frac{2\rho \Phi^{-1}(u_0) \Phi^{-1}(u_1) - \rho^2 \Phi^{-1}(u_1)^2 - \Phi^{-1}(u_0)^2}{2(1 - \rho^2)} \Bigr) \;,
\end{align*}
with joint copula density
\begin{align}
\label{eqn:gaussian-copula-density}
c_{\mathrm{G}}^{\rho}(u_1, u_0) := \frac{1}{\sqrt{1-\rho^2}} \exp\Bigl( \frac{2\rho \Phi^{-1}(u_0) \Phi^{-1}(u_1) - \rho^2 \Phi^{-1}(u_1)^2 - \rho^2 \Phi^{-1}(u_0)^2}{2(1 - \rho^2)} \Bigr) \;.
\end{align}

(ii) The Bregman-Sinkhorn copula reduces to the comonotonic copula in the limit when $\rho \to 1$ and $\epsilon \to 0$. In this case, $\psi^\ast_1 = G^{-1} \circ F$ that achieves the Fréchet-Hoeffding copula, namely
\begin{align*}
c_{\mathrm{FH}}(u_1 \mid U_0 = u_0) := \delta_{u_0}(u_1) \;.
\end{align*}

\subsection{Rank Stickiness}

We now study the rank stickiness property of the nonparametric Bregman-Sinkhorn copula. The conditional rank distribution $G(Y_1) \mid F(Y_0) = u_0$ admits an analytic formula via the Bregman divergence representation in Proposition~\ref{prop:bregman-divergence-representation}. We characterize the rank distortion of the Bregman-Sinkhorn copula compared to the comonotonic copula, in the limiting case $\rho \to 1$ and $\epsilon \to 0$.

\begin{proposition}[Rank stickiness]
  \label{prop:rank-stickiness}
  Consider the Bregman-Sinkhorn copula $(G(Y_1), F(Y_0)) \sim c_{\mathrm{B}}^{\rho, \epsilon}$ defined in \cref{eqn:conditional-bregman-sinkhorn-copula} with $\epsilon \in [0, 1)$ and $\rho \in (0, 1]$. Assume that potential $\psi_\rho$ is strictly convex and $\mu, \nu$ have strictly positive and smooth densities supported on a bounded interval.
  
  Then the following rank stickiness properties hold for any fixed $u_0 \in [0, 1]$ and $\rho \in (0, 1]$ and a sequence $\epsilon_n \to 0$:
  \begin{align*}
  \E[ G(Y_1) \mid F(Y_0) = u_0] &= G \circ \dot\psi^\ast_{\rho} \circ F^{-1} (u_0) + o_n(1) \\
  \mathop{Var}[ G(Y_1) \mid F(Y_0) = u_0 ] &= o_n(1)
  \end{align*}
\end{proposition}
\begin{remark}
  We call this rank stickiness for the following reason: when $\rho=1$ and $\epsilon_n$ small, the conditional rank of $G(Y_1) \mid F(Y_0) = u_0$ is concentrated around $G \circ \dot\psi^\ast_{1} \circ F^{-1} (u_0) = G \circ (G^{-1} \circ F) \circ F^{-1} (u_0) = u_0$, with bias and variance both of the order $o_n(1)$.

  Together with Proposition~\ref{prop:rank-violations}, this shows that the Bregman-Sinkhorn copula exhibits rank stickiness in the sense that the conditional rank of $Y_1$ given $Y_0$ is near the true rank when $\rho \to 1$, while still allowing for rank violations with non-vanishing probability when $\rho \in (0, 1)$.
\end{remark}

\section{Identification and Inference}
\label{sec:identification-inference}

\subsection{Identification}

We now prove that the Bregman-Sinkhorn copula is nonparametrically identifiable and unique under the following assumptions, without imposing any parametric restrictions on the marginals $\mu, \nu$ or the copula $\pi_\rho$. We establish the result via an asymmetric contraction argument suited for the Sinkhorn system \eqref{eqn:sinkhorn-system}.

\begin{assumption}
  \label{asmp:no-atoms}
\begin{enumerate}
  \item[(i)] $\cY \subset \R^d$ is compact.
  \item[(ii)] $\mu, \nu$ are absolutely continuous w.r.t. the Lebesgue measure and have strictly positive densities.
\end{enumerate}
\end{assumption}

\begin{theorem}[Nonparametric identifiability]
  \label{thm:identifiability}
Under Assumption \ref{asmp:no-atoms}, the $\rho$-Bregman-Sinkhorn copula $\pi_\rho$ that solves \cref{eqn:pi-rho} is uniquely identifiable given the marginals $\mu, \nu$, for any $\rho\in (0, 1]$.
\end{theorem}

\begin{proof}[Proof of Theorem~\ref{thm:identifiability}]
The case $\rho = 1$ follows from Brenier's theorem \citep{brenier1991polar}. We hence consider $\rho \in (0,1)$ and denote $\epsilon := \epsilon(\rho) > 0$. By standard Sinkhorn theory for strictly positive continuous costs on compact spaces, there exists a bounded continuous pair $(\phi_\rho, \psi_\rho)$ solving
\begin{align*}
\phi_\rho(y_0) &= \epsilon \log \int_{\cY} \exp\Bigl( \frac{\langle y_0, y_1 \rangle - \psi_{\rho}(y_1)}{\epsilon} \Bigr) \dd \nu(y_1) \;, \\
\psi_\rho(y_1) &= \epsilon \log \int_{\cY} \exp\Bigl( \frac{\langle y_1, y_0 \rangle - \phi_\rho(y_0)}{\epsilon} \Bigr) \dd \mu(y_0) \;.
\end{align*}
It therefore suffices to prove uniqueness of such a pair up to additive constants: namely, the pair $(\phi_\rho, \psi_\rho)$ is the unique solution to the above system up to the transformation $(\phi_\rho, \psi_\rho) \mapsto (\phi_\rho - c, \psi_\rho + c)$ for any $c \in \R$.

Let $(\phi^0, \psi^0)$ and $(\phi^1, \psi^1)$ be two bounded continuous solutions. Define
\begin{align*}
h(y_1) := \psi^1(y_1) - \psi^0(y_1) \;, \qquad M_\psi = \|h\|_\infty := \sup_{y_1 \in \cY} |h(y_1)| \;.
\end{align*}
We claim $h$ must be constant; suppose for contradiction that $h$ is non-constant.

For $t \in [0,1]$, define
\begin{align*}
\psi^t := (1-t)\psi^0 + t\psi^1 \;,
\qquad
\phi^t(y_0) := \epsilon \log \int_{\cY} \exp\Bigl( \frac{\langle y_0, y_1 \rangle - \psi^t(y_1)}{\epsilon} \Bigr) \dd \nu(y_1) \;.
\end{align*}
Since $\psi^0,\psi^1$ are bounded and continuous and $\cY$ is compact, differentiation is justified, and for each $y_0 \in \cY$,
\begin{align*}
\phi^1(y_0) - \phi^0(y_0)
= \int_0^1 \frac{\partial \phi^t(y_0)}{\partial t} \dd t
= -\int_0^1 \int_{\cY} h(y_1) w_t(y_0,\dd y_1) \dd t \;,
\end{align*}
where
\begin{align*}
w_t(y_0,\dd y_1)
:= \frac{\exp\bigl( (\langle y_0,y_1\rangle - \psi^t(y_1))/\epsilon \bigr)}{\int_{\cY} \exp\bigl( (\langle y_0,y'_1\rangle - \psi^t(y'_1))/\epsilon \bigr) \dd \nu(y'_1)} \dd \nu(y_1)
\end{align*}
is a probability measure on $\cY$.

Because $h$ is continuous and non-constant on the compact set $\cY$, there exist $\delta_\psi \in (0,1)$ and a nonempty open set
\begin{align*}
A_\psi := \{y_1 \in \cY: |h(y_1)| < (1-\delta_\psi) M_\psi\}
\end{align*}
with $\nu(A_\psi) > 0$. Next, let
\begin{align*}
R := \sup_{y_0,y_1 \in \cY} |\langle y_0,y_1\rangle| < \infty \;,
\qquad
B_\psi := \max\{\|\psi^0\|_\infty, \|\psi^1\|_\infty\} \;.
\end{align*}
Then for every $t \in [0,1]$ and $y_0,y_1 \in \cY$,
\begin{align*}
\exp\Bigl( \frac{-R-B_\psi}{\epsilon} \Bigr)
\leq
\exp\Bigl( \frac{\langle y_0,y_1\rangle - \psi^t(y_1)}{\epsilon} \Bigr)
\leq
\exp\Bigl( \frac{R+B_\psi}{\epsilon} \Bigr) \;.
\end{align*}
Hence, for every measurable $A \subset \cY$,
\begin{align*}
w_t(y_0,A)
\geq
\exp\Bigl( -\frac{2(R+B_\psi)}{\epsilon} \Bigr) \nu(A)
=: c_\psi \nu(A) \;,
\end{align*}
uniformly in $y_0$ and $t$. Therefore,
\begin{align*}
\bigl|\phi^1(y_0)-\phi^0(y_0)\bigr|
&\leq \int_0^1 \int_{\cY} |h(y_1)| w_t(y_0,\dd y_1) \dd t \\
&\leq \int_0^1 (1-\delta_\psi)M_\psi  w_t(y_0,A_\psi)  + M_\psi  w_t(y_0,\cY \setminus A_\psi) \dd t \\
&\leq \bigl( 1 - \delta_\psi c_\psi \nu(A_\psi) \bigr) M_\psi \;.
\end{align*}
Taking the supremum over $y_0$ gives
\begin{align}
  \label{eqn:phi-contraction}
\|\phi^1 - \phi^0\|_\infty \leq \alpha_\psi \|\psi^1 - \psi^0\|_\infty \;,
\qquad
\alpha_\psi := 1 - \delta_\psi c_\psi \nu(A_\psi) < 1 \;.
\end{align}

The same interpolation applied to the second fixed-point equation gives, for each $y_1 \in \cY$,
\begin{align*}
\psi^1(y_1) - \psi^0(y_1) = -\int_0^1 \int_{\cY} \bigl(\phi^1(y_0) - \phi^0(y_0)\bigr)\, v_t(y_1, \dd y_0)\, \dd t \;,
\end{align*}
where $v_t(y_1, \dd y_0)$ is the analogous probability measure with $(\mu,\phi)$ in place of $(\nu,\psi)$. Since $v_t(y_1,\cdot)$ is a probability measure,
\begin{align*}
\|\psi^1 - \psi^0\|_\infty \leq \|\phi^1 - \phi^0\|_\infty \;.
\end{align*}
Combining with \eqref{eqn:phi-contraction} yields
\begin{align*}
\|\psi^1 - \psi^0\|_\infty \leq \alpha_\psi \|\psi^1 - \psi^0\|_\infty \;.
\end{align*}
Since $\alpha_\psi < 1$ and $M_\psi > 0$ ($\psi^0 \neq \psi^1$), this is a contradiction. Therefore $h$ must be constant, i.e., $\psi^1 - \psi^0 \equiv c$ for some $c \in \R$.

Substituting $\psi^1 = \psi^0 + c$ into the first fixed-point equation gives, for every $y_0 \in \cY$,
\begin{align*}
\phi^1(y_0) = \epsilon \log \int_{\cY} \exp\Bigl( \frac{\langle y_0,y_1\rangle - \psi^0(y_1) - c}{\epsilon} \Bigr) \dd \nu(y_1) = \phi^0(y_0) - c \;.
\end{align*}
Hence any two solutions differ by additive constants of opposite sign.

Finally, the coupling defined by \cref{eqn:sinkhorn-copula} is invariant under the transformation $(\phi,\psi) \mapsto (\phi-c, \psi+c)$. Therefore the resulting coupling $\pi_\rho$ is uniquely determined by $(\mu,\nu)$. This proves identifiability for $\rho \in (0,1)$, and the case $\rho=1$ was treated above.
\end{proof}

\subsection{Limit Theorem}

Suppose $(Y_{0,i})_{i=1}^n \stackrel{i.i.d.}{\sim} \mu$ and $(Y_{1,i})_{i=1}^n \stackrel{i.i.d.}{\sim} \nu$ are independent samples. Let $\mu_n, \nu_n$ be the corresponding empirical measures, $\mu_n := \frac{1}{n} \sum_{i=1}^n \delta_{Y_{0,i}}$ and $\nu_n := \frac{1}{n} \sum_{i=1}^n \delta_{Y_{1,i}}$. Let $F_n, G_n$ be the empirical CDFs of $\mu_n, \nu_n$, and let $(\phi_n, \psi_n)$ be the empirical Sinkhorn potentials constructed from $\mu_n, \nu_n$ as in \cref{eqn:sinkhorn-system}.

We establish central limit theorems (CLT) for the empirical potentials $\phi_n \oplus \psi_n (y_0, y_1) := \phi_n(y_0) + \psi_n(y_1)$ constructed from the marginal empirical measures $\mu_n, \nu_n$. These complement the identifiability result of \cref{thm:identifiability}: the Bregman-Sinkhorn copula is not only nonparametrically identified but also estimable at the parametric $\sqrt{n}$-rate, with a Gaussian process limit distribution.

The following two theorems use different proof strategies. For $\rho = 1$, the empirical rank-preserving map $G_n^{-1} \circ F_n$ is Hadamard differentiable as a composition of quantile functional, and the result follows from the functional delta method \citep[Theorem 3.9.4]{van1996weak}. For $\rho \in (0,1)$, the Sinkhorn dual potentials are defined implicitly by the fixed-point system \eqref{eqn:sinkhorn-system}; the main technical work is to establish (i) the Fr\'echet derivative of the Sinkhorn operator with respect to $h = \phi \oplus \psi$, and (ii) that the linearized operator (the derivative) is a bounded linear isomorphism on a gauge-centered Banach subspace, using a variant of the argument in Theorem~\ref{thm:identifiability}. The CLT then follows by applying the implicit function theorem in Banach spaces. We defer the proofs to the appendix.

We first establish the limit theorem for the rank-preserving case $\rho=1$, where $G_n^{-1} \circ F_n$ and $G^{-1} \circ F$ denote the empirical and population couplings, respectively.

\begin{theorem}[CLT for empirical approximation, $\rho = 1$]
  \label{thm:limit-theorem-rho-1}
  Let $\cY \subset \R$ be a bounded interval. Denote $\ell^\infty(\cY)$ as the space of bounded measurable functions. Suppose $(Y_{0,i})_{i=1}^n \stackrel{i.i.d.}{\sim} \mu$ and $(Y_{1,i})_{i=1}^n \stackrel{i.i.d.}{\sim} \nu$ are independent samples. Assume $F$ and $G$ are continuously differentiable on $\cY$ with densities $f,g$, and
  \begin{align*}
    0 < \inf_{y \in \cY} f(y) \leq \sup_{y \in \cY} f(y) < \infty \;, \qquad
    0 < \inf_{y \in \cY} g(y) \leq \sup_{y \in \cY} g(y) < \infty \;.
  \end{align*}
  Then, in $\ell^\infty(\cY)$, we have the following limit:
\begin{align*}
\sqrt{n} \, \bigl( G_n^{-1} \circ F_n(y) - G^{-1} \circ F(y) \bigr) \rightsquigarrow \frac{\sqrt{2}\,\bB(F(y))}{g(G^{-1}(F(y)))} \;,
\end{align*}
for a standard Brownian bridge $\bB$.
\end{theorem}

Now we establish the limit theorem for the $\rho \in (0, 1)$ case. We prove for the general multivariate case $\cY \subset \R^d$ with $d \geq 1$ under Assumption~\ref{asmp:no-atoms}. Limit theorems for entropic OT have been established in the recent literature \citep{gonzalez2022weak, goldfeld2024limit}, here we provide a self-contained analysis.

\begin{theorem}[CLT for empirical approximation, $\rho \in (0,1)$]
  \label{thm:limit-theorem-rho-less-than-1}
  Let Assumption \ref{asmp:no-atoms} hold and consider $\rho \in (0,1)$. Denote $\cC(\cY)$ as the space of continuous functions. Suppose $(Y_{0,i})_{i=1}^n \stackrel{i.i.d.}{\sim} \mu$ and $(Y_{1,i})_{i=1}^n \stackrel{i.i.d.}{\sim} \nu$ are independent samples. Let $(\phi_\rho,\psi_\rho)$ be Sinkhorn dual potentials for $(\mu,\nu)$, and $(\phi_n,\psi_n)$ the corresponding empirical dual potentials for $(\mu_n,\nu_n)$, as in \cref{eqn:sinkhorn-system}.

  Then, in $C(\cY) \oplus C(\cY) \subset C(\cY \times \cY)$,
  \begin{align*}
    \sqrt{n} \, \bigl( \phi_n \oplus \psi_n - \phi_\rho \oplus \psi_\rho \bigr) \rightsquigarrow \bZ_\rho := \bZ_{0, \rho} \oplus \bZ_{1, \rho} \;,
  \end{align*}
  where $(\bZ_{0,\rho}, \bZ_{1,\rho})$ is a jointly mean-zero Gaussian process in $C(\cY) \times C(\cY)$.
\end{theorem}

\section{Application to Causal Inference}

This section applies the Bregman-Sinkhorn copula to three problems in causal inference: (i) Section~\ref{sec:application} on estimation of the treatment effect distribution (TED), (ii) Section~\ref{sec:variance-functional} on estimation of the variance of the ATE estimator under rank violations, and (iii) Section~\ref{sec:covariate-adjustment} on extending the framework to observational studies via covariate adjustment.

We adopt the potential outcomes framework. Let $(Y_0, Y_1) \in \cY \times \cY$ denote the potential outcomes under control and treatment, respectively, and let $T \in \{0, 1\}$ be a binary treatment indicator. The observed outcome is $Y_i = T_i Y_{1,i} + (1-T_i) Y_{0,i}$, and we observe $n$ i.i.d.\ copies $\{(Y_i, T_i)\}_{i=1}^n$. The object of interest is the treatment effect distribution (TED) $\cL(Y_1 - Y_0)$ and functionals thereof, including the variance of the Horvitz-Thompson estimator of the ATE.

\subsection{Treatment Effect Distribution}
\label{sec:application}

We first work in the experimental setting where $T \sim \mathrm{Bern}(1/2)$ is independent of $(Y_0, Y_1)$. The joint distribution of potential outcomes is assumed to follow the Bregman-Sinkhorn copula: $(Y_0, Y_1) \sim \pi_\rho \in \Pi(\mu, \nu)$, where $\mu$ and $\nu$ are the marginal distributions of $Y_0$ and $Y_1$, respectively.

The Bregman-Sinkhorn copula $\widehat{\pi}_\rho$ is estimated by solving the Sinkhorn system with empirical marginals
\begin{align*}
\mu_n(y_0) = \frac{1}{n_{\sC}} \sum_{i \in \sC} \delta_{Y_i \mid T_i = 0}(y_0) \;, \qquad \nu_n(y_1) = \frac{1}{n_{\sT}} \sum_{i \in \sT} \delta_{Y_i \mid T_i = 1}(y_1) \;,
\end{align*}
where $\sC := \{i : T_i = 0\}$ and $\sT := \{i : T_i = 1\}$ are the control and treatment groups, with sizes $n_{\sC}$ and $n_{\sT}$, respectively. Since $T \indep (Y_0, Y_1)$, we have $\cL(Y \mid T=0) = \cL(Y_0)$ and $\cL(Y \mid T=1) = \cL(Y_1)$, so the empirical marginals $\mu_n, \nu_n$ are consistent for $\mu, \nu$.

The Sinkhorn system \eqref{eqn:sinkhorn-system} is solved iteratively until convergence, yielding empirical dual potentials $(\phi_n, \psi_n)$. The empirical copula $\widehat{\pi}_\rho$ is then obtained by substituting $(\phi_n, \psi_n)$ and $(\mu_n, \nu_n)$ into \cref{eqn:sinkhorn-copula}.

The TED is then estimated by pushing $\widehat{\pi}_\rho$ forward through the map $(y_0, y_1) \mapsto y_1 - y_0$:
$$
\widehat{\cL_\rho(Y_1 - Y_0)} := (y_1 - y_0) \sharp \widehat{\pi}_\rho \in \sP(\cY) \;.
$$

\begin{figure}[htbp]
  \centering
  %% Row 1: ATE distribution
  \begin{subfigure}{0.49\linewidth}
    \includegraphics[width=\linewidth]{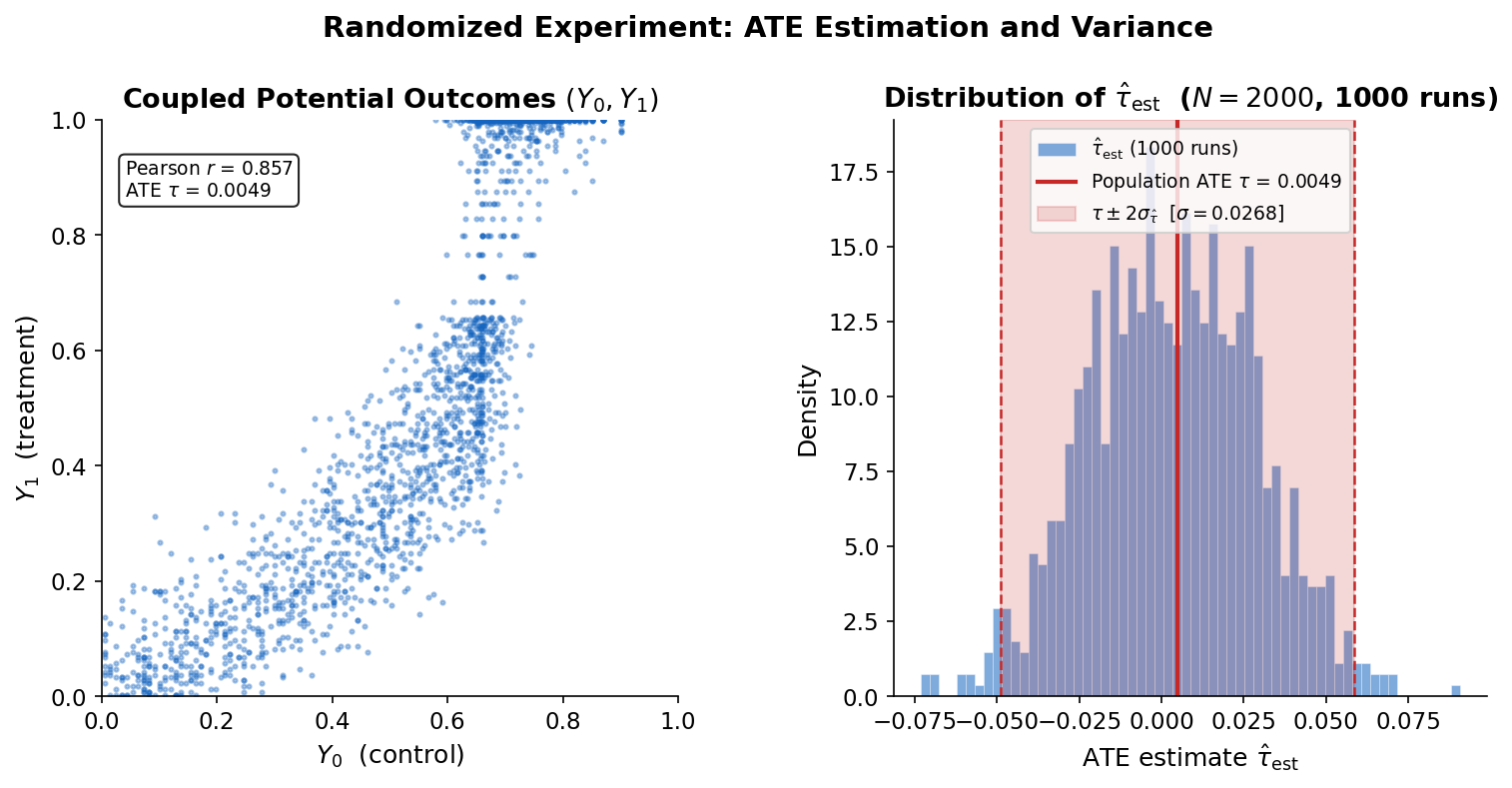}
    \caption{Bregman-Sinkhorn copula DGP, ATE distribution.}
    \label{fig:ate-sinkhorn}
  \end{subfigure}\hfill
  \begin{subfigure}{0.49\linewidth}
    \includegraphics[width=\linewidth]{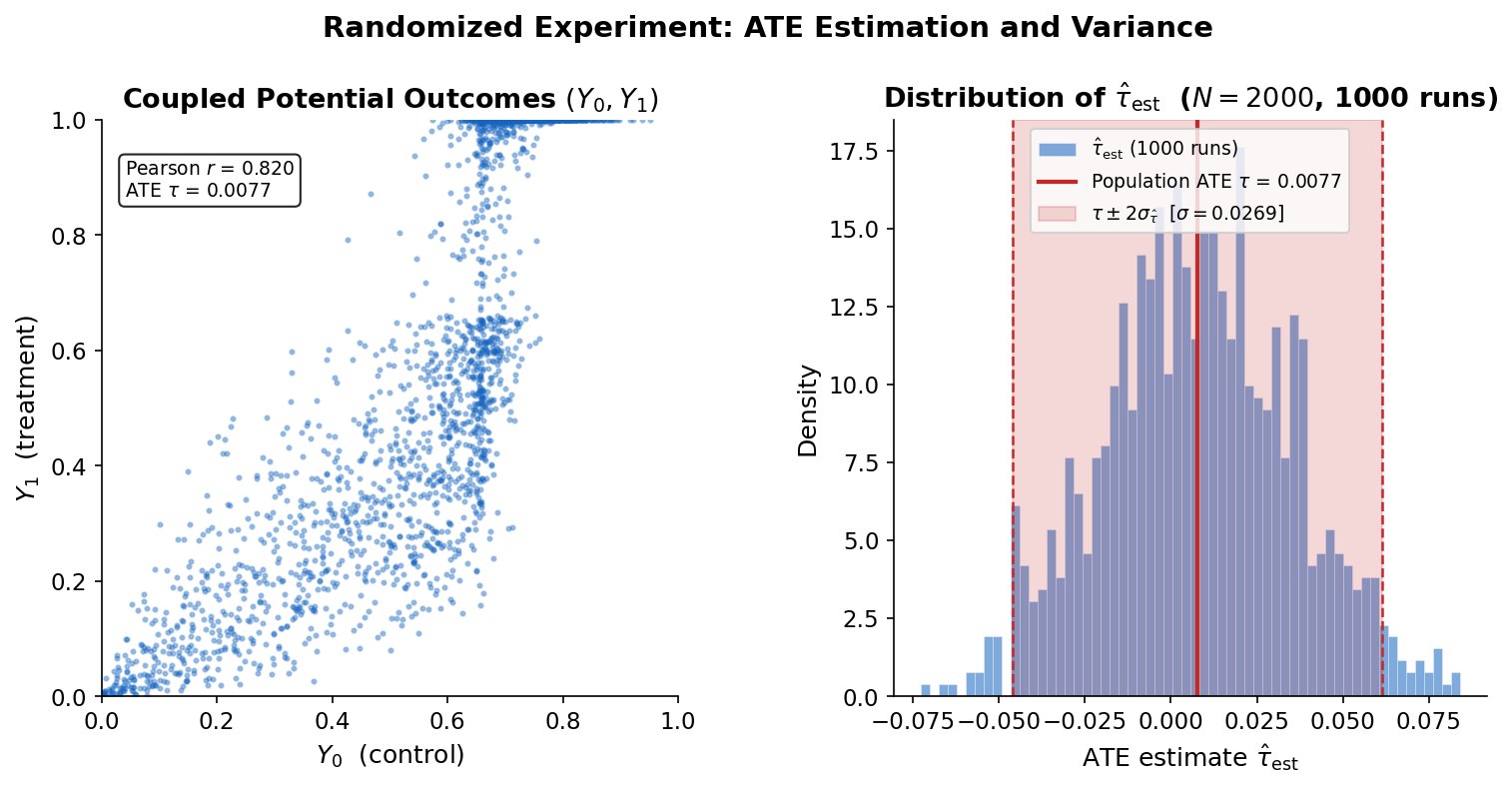}
    \caption{Gaussian copula DGP, ATE distribution.}
    \label{fig:ate-gaussian}
  \end{subfigure}\\[4pt]
  %% Row 2: TED imputed, Sinkhorn
  \begin{subfigure}{\linewidth}
    \includegraphics[width=\linewidth]{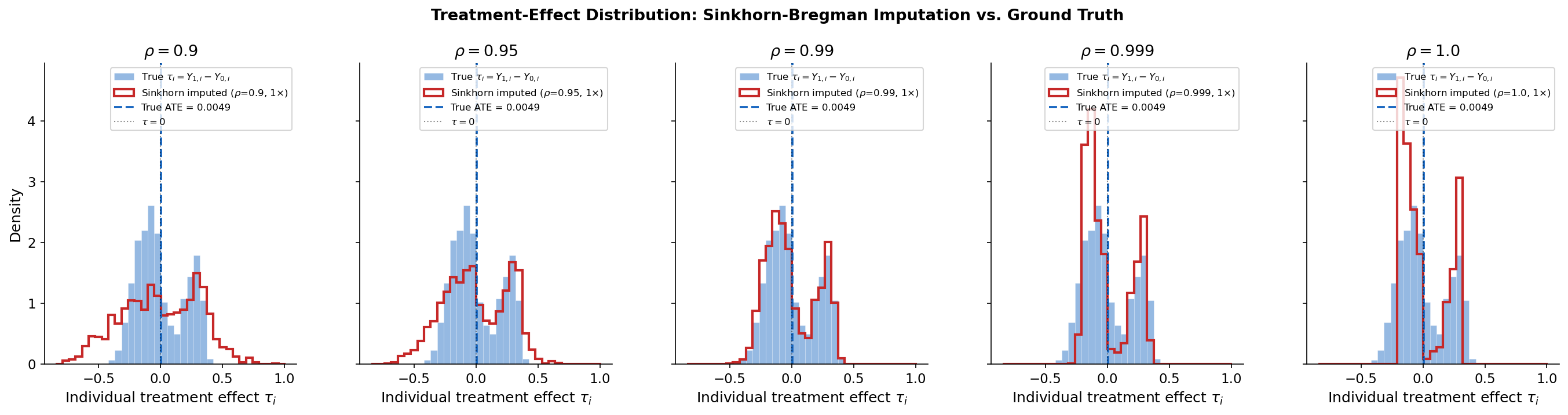}
    \caption{Imputed TED $\widehat{\mathcal{L}_\rho(Y_1-Y_0)}$, Bregman-Sinkhorn copula DGP, varying $\rho$.}
    \label{fig:ted-sinkhorn}
  \end{subfigure}\\[4pt]
  %% Row 3: TED imputed, Gaussian
  \begin{subfigure}{\linewidth}
    \includegraphics[width=\linewidth]{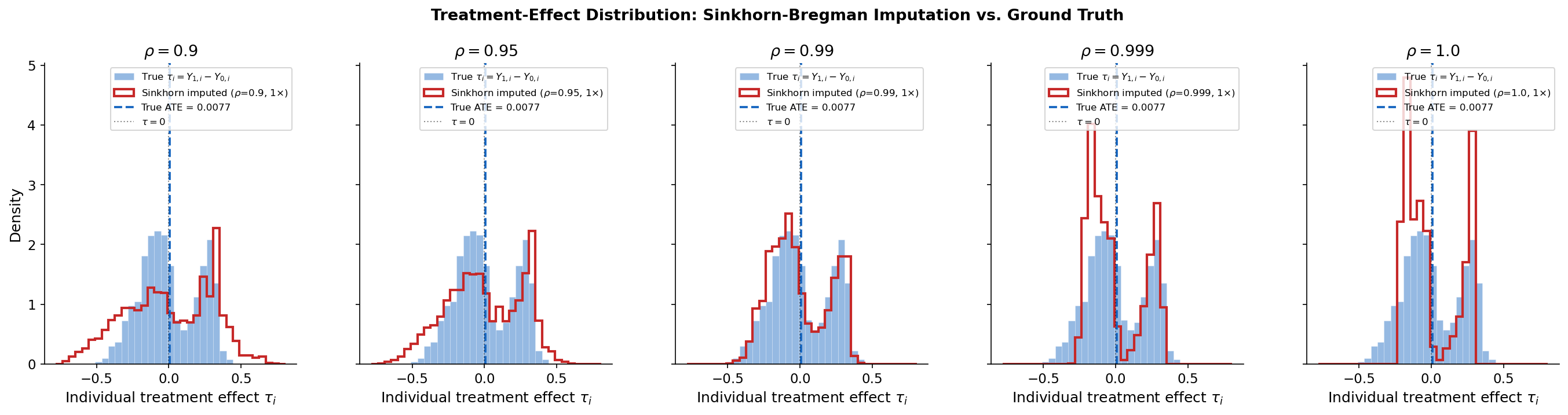}
    \caption{Imputed TED $\widehat{\mathcal{L}_\rho(Y_1-Y_0)}$, Gaussian copula DGP, varying $\rho$.}
    \label{fig:ted-gaussian}
  \end{subfigure}
  \caption{\small Numerical illustrations. \textit{Row 1}: ATE estimator distribution under each copula DGP. \textit{Rows 2--3}: imputed treatment effect distributions across stickiness levels $\rho$ under the Bregman-Sinkhorn copula DGP (row 2) and Gaussian copula DGP (row 3).}
  \label{fig:copula-illustrations}
\end{figure}

\Cref{fig:copula-illustrations} illustrates the imputed TED $\widehat{\cL_\rho(Y_1 - Y_0)}$ under two data-generating processes (DGPs): a Bregman-Sinkhorn copula DGP with $(F(Y_0), G(Y_1)) \sim c_{\mathsf{S}}^{\rho_1=0.99}$ as in \cref{eqn:sinkhorn-copula-density}, and a Gaussian copula DGP with $(F(Y_0), G(Y_1)) \sim c_{\mathsf{G}}^{\rho_2=0.90}$ as in \cref{eqn:gaussian-copula-density}. Both DGPs exhibit non-trivial rank violations, as evidenced by \cref{fig:ate-sinkhorn,fig:ate-gaussian}. The top row of \cref{fig:copula-illustrations} displays the sampling distribution of the ATE estimator $\widehat{\tau} = n_{\sT}^{-1} \sum_{i \in \sT} Y_i - n_{\sC}^{-1} \sum_{i \in \sC} Y_i$ across 1000 runs of the treatment assignment $T$, holding the potential outcomes fixed; the population ATE is a small positive constant in both cases.

For each DGP, $n = 2000$ i.i.d.\ draws of $(Y_0, Y_1, T)$ are generated, and the imputed TED $\widehat{\cL_\rho(Y_1 - Y_0)}$ is computed for stickiness levels $\rho \in \{0.90, 0.95, 0.99, 0.999, 1.00\}$. Rows 2--3 of \cref{fig:copula-illustrations} compare the imputed TED $\widehat{\cL_\rho(Y_1 - Y_0)}$ against the true, unobserved $\cL_\rho(Y_1 - Y_0)$ for each DGP across these values of $\rho$.

Both DGPs are calibrated so that the population ATE $\tau := n^{-1}\sum_{i=1}^n (Y_{1,i} - Y_{0,i})$ is small and positive, yet the TED is bimodal: approximately $60\%$ of units experience a negative individual treatment effect and $40\%$ experience a positive one. The imputed TED $\widehat{\cL_\rho(Y_1 - Y_0)}$ recovers this bimodal structure across all stickiness levels $\rho$, and a well-chosen $\rho$ (e.g., $\rho = 0.99$) closely approximates the true population TED. By contrast, the ATE estimator fails to capture this heterogeneity and, under the couplings considered here, exhibits substantial variance, rendering ATE-based inference uninformative.

The comonotonic coupling ($\rho = 1$) underestimates treatment effect heterogeneity by producing an overly concentrated TED. The Bregman-Sinkhorn copula with $\rho < 1$ permits rank violations, interpolating between the comonotonic coupling ($\rho = 1$, which over-shrinks the TED) and the independence coupling ($\rho \to 0$, which over-disperses it), and recovering the full heterogeneity at intermediate values of $\rho$.

\subsection{Variance Functional of ATE Estimator}
\label{sec:variance-functional}
The Horvitz-Thompson estimator of the ATE and its exact variance are
\begin{align*}
  \widehat{\tau} := \frac{2}{n} \sum_{i=1}^n Y_i T_i -  \frac{2}{n} \sum_{i=1}^n  Y_i (1-T_i)  \;, \qquad
  \mathop{Var}(\widehat{\tau}) = \frac{1}{n^2} \sum_{i=1}^n (Y_{0,i} + Y_{1,i})^2 \;.
\end{align*}
Since $\mathop{Var}(\widehat{\tau})$ depends on the unobserved joint distribution of $(Y_0, Y_1)$, it must be bounded or imputed under structural assumptions. Under the structural assumption $(Y_0, Y_1) \sim \pi_\rho$, the Bregman-Sinkhorn copula, the conditional moment formulas of Proposition~\ref{prop:conditional-moments} yield a tractable expression for imputation.

\begin{proposition}[Variance functional]
\label{prop:variance-functional}
Under the Bregman-Sinkhorn copula $(Y_0, Y_1) \sim \pi_\rho$,
\begin{align*}
\E\bigl[(Y_{0,i} + Y_{1,i})^2\bigr] =
 \E\bigl[ \bigl(Y_{0,i} + \dot\phi_\rho(Y_{0,i})\bigr)^2 + \epsilon(\rho)\,\ddot\phi_\rho(Y_{0,i}) \bigr] =
 \E\bigl[ \bigl(Y_{1,i} + \dot\psi_\rho(Y_{1,i})\bigr)^2 + \epsilon(\rho)\,\ddot\psi_\rho(Y_{1,i}) \bigr] \;.
\end{align*}
The resulting plug-in estimator of $\mathop{Var}(\widehat{\tau})$ is
\begin{align*}
  \widehat{\mathop{Var}}_{\mathsf{SB}}(\widehat{\tau}) := \frac{1}{n^2} \sum_{i=1}^n \Bigl\{ \bigl[Y_i + \dot\phi_\rho(Y_i)\bigr]^2 + \epsilon(\rho)\,\ddot\phi_\rho(Y_i) \Bigr\} T_i
  + \Bigl\{ \bigl[Y_i + \dot\psi_\rho(Y_i)\bigr]^2 + \epsilon(\rho)\,\ddot\psi_\rho(Y_i) \Bigr\} (1-T_i) \;.
\end{align*}
\end{proposition}
\begin{proof}
  Apply Proposition~\ref{prop:conditional-moments} and the law of total expectation.
\end{proof}

In practice, $\dot\phi_\rho$ and $\ddot\phi_\rho$ are computed from the empirical dual potentials $\phi_n, \psi_n$ obtained by running the Sinkhorn algorithm on $\mu_n, \nu_n$. By Proposition~\ref{prop:conditional-moments}, the empirical counterparts are (with $\dot\psi_\rho$ and $\ddot\psi_\rho$ defined analogously):
\begin{align*}
\widehat{\dot\phi_\rho}(y_0) =  \frac{\E_{Y_1 \sim \nu_n} \Bigl[ Y_1 \exp \Bigl( \frac{y_0 Y_1 - \psi_n(Y_1)}{\epsilon} \Bigr) \Bigr]}{\E_{Y_1 \sim \nu_n} \Bigl[ \exp \Bigl( \frac{y_0 Y_1 - \psi_n(Y_1)}{\epsilon} \Bigr) \Bigr]} \;, \quad
\widehat{\ddot\phi_\rho}(y_0) =  \frac{\E_{Y_1 \sim \nu_n} \Bigl[ \bigl(Y_1 - \widehat{\dot\phi_\rho}(y_0)\bigr)^2 \exp \Bigl( \frac{y_0 Y_1 - \psi_n(Y_1)}{\epsilon} \Bigr) \Bigr]}{\E_{Y_1 \sim \nu_n} \Bigl[ \exp \Bigl( \frac{y_0 Y_1 - \psi_n(Y_1)}{\epsilon} \Bigr) \Bigr]} \;.
\end{align*}

We benchmark $\widehat{\mathop{Var}}_{\mathsf{SB}}$ against two competitors. The first is the Fréchet-Hoeffding bound, corresponding to the special case $\rho = 1$ (rank preservation):
\begin{align*}
 \widehat{\mathop{Var}}_{\mathsf{FH}}(\widehat{\tau}) := \frac{1}{n^2} \sum_{i=1}^n \bigl(Y_i + \widehat{G}^{-1} \circ \widehat{F}(Y_i)\bigr)^2 T_i + \bigl(Y_i + \widehat{F}^{-1} \circ \widehat{G}(Y_i)\bigr)^2 (1-T_i) \;.
\end{align*}
The second is the classical upper bound of \citet{neyman1923application}, which relies only on the marginal variances:
\begin{align*}
  \widehat{\mathop{Var}}_{\mathsf{N}}(\widehat{\tau}) := \frac{1}{n}\Bigl( \widehat{V}_0 + \widehat{V}_1 + 2\sqrt{\widehat{V}_0\,\widehat{V}_1} \Bigr) \;, \quad \widehat{V}_0 := \frac{2}{n}\sum_{i=1}^n Y_i^2(1-T_i) \;, \quad \widehat{V}_1 := \frac{2}{n}\sum_{i=1}^n Y_i^2 T_i \;.
\end{align*}

\begin{figure}[htbp]
  \centering
  % Row 1: Copula DGP
  \includegraphics[width=\linewidth]{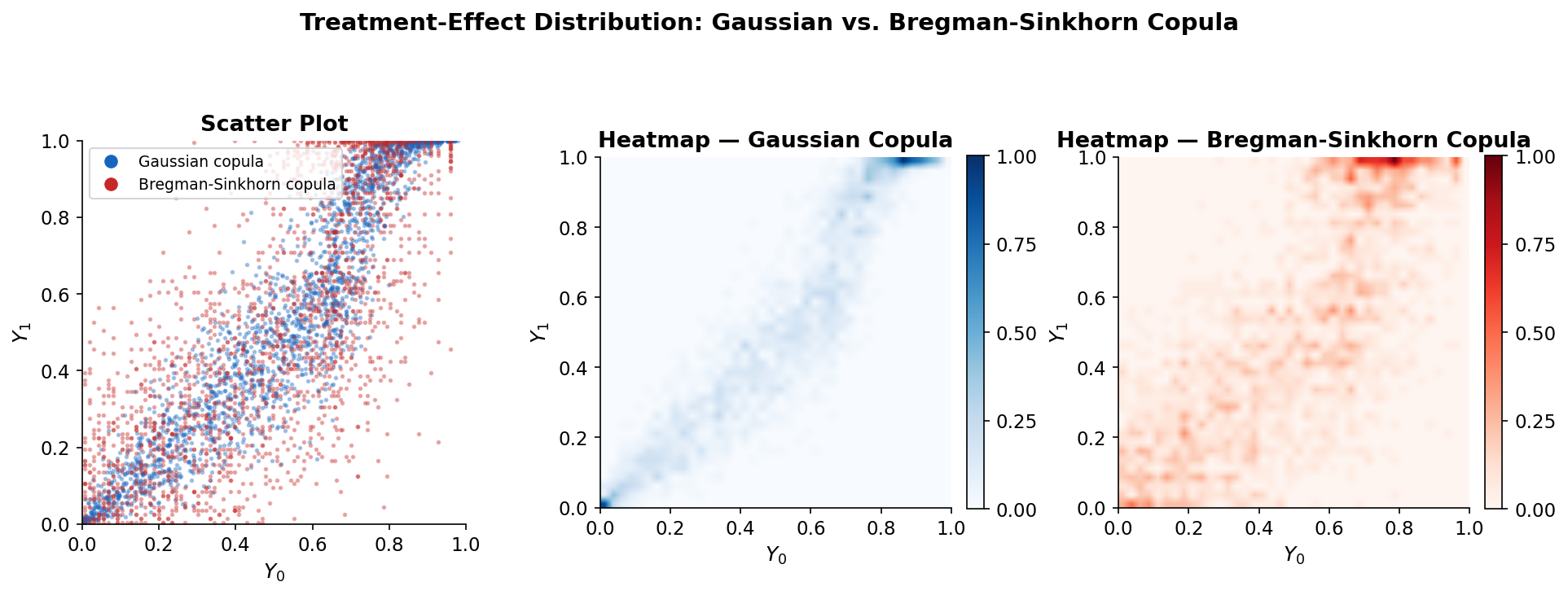}\\[6pt]
  % Row 2: Variance bounds — Sinkhorn (left) and Gaussian (right)
  \begin{subfigure}[t]{0.49\linewidth}
    \centering
    \includegraphics[width=\linewidth]{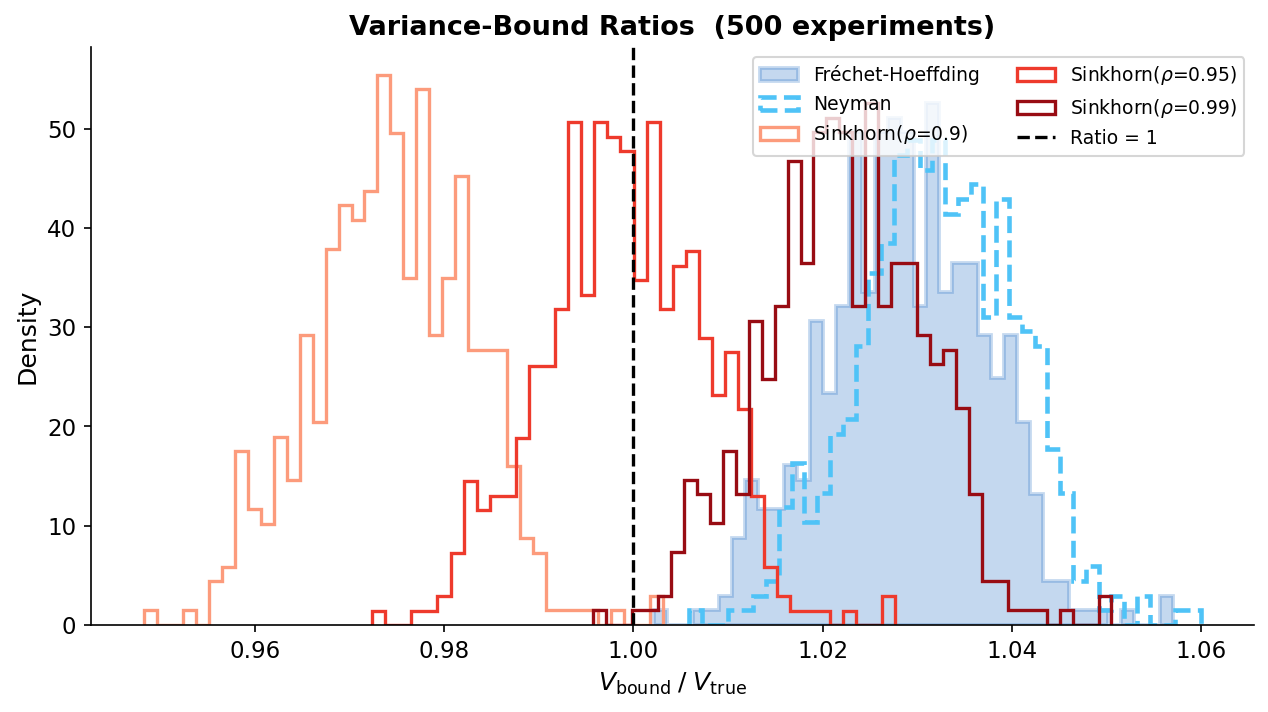}
    \caption{Bregman--Sinkhorn copula DGP.}
    \label{fig:variance-ate-bounds-sinkhorn}
  \end{subfigure}
  \hfill
  \begin{subfigure}[t]{0.49\linewidth}
    \centering
    \includegraphics[width=\linewidth]{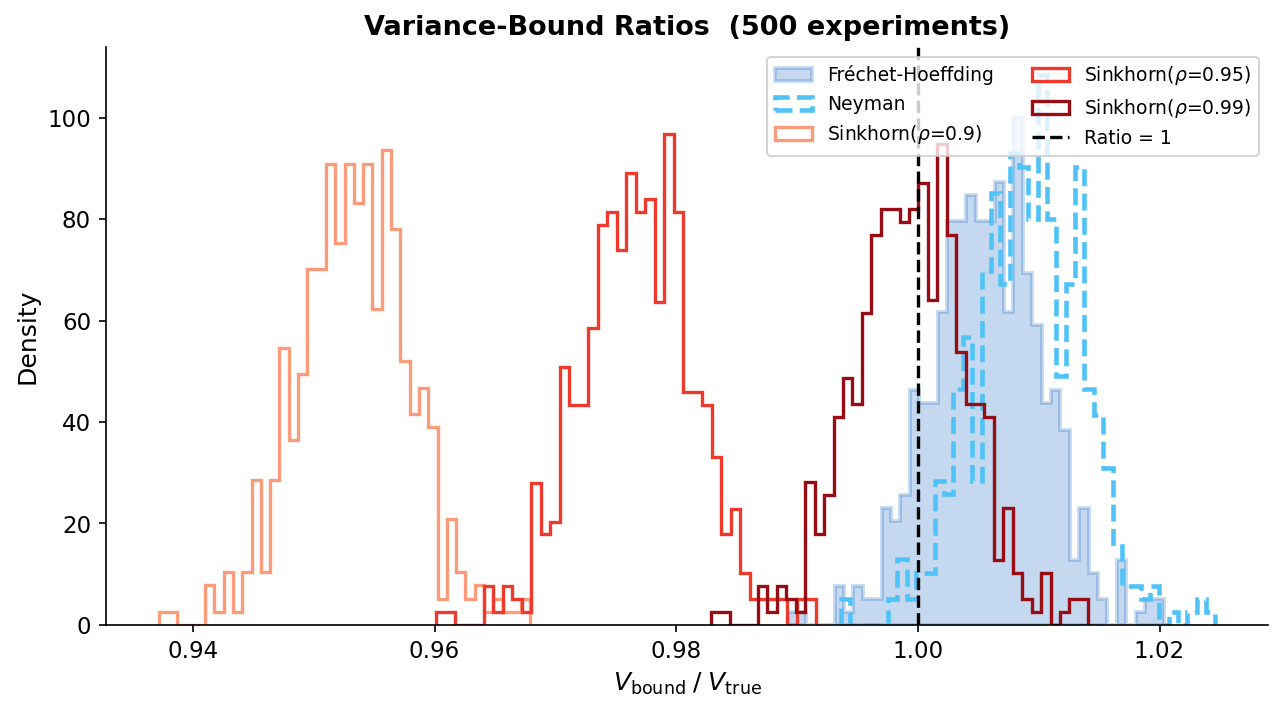}
    \caption{Gaussian copula DGP.}
    \label{fig:variance-ate-bounds-gaussian}
  \end{subfigure}
  \caption{\small Variance bounds of the ATE: $\widehat{\mathrm{Var}}_{\mathsf{FH}}$, $\widehat{\mathrm{Var}}_{\mathsf{N}}$ versus $\widehat{\mathrm{Var}}_{\mathsf{SB}}$ by varying the stickiness parameter $\rho \in \{0.9, 0.95, 0.99\}$, under the Bregman--Sinkhorn copula DGP (left) and the Gaussian copula DGP (right). The top row visualizes the copula DGPs and rank violations by plotting the unobserved $(Y_0, Y_1)$ pairs. The bottom row compares the variance bounds against the true variance across 500 independent draws of $T$.}
  \label{fig:variance-ate-bounds}
\end{figure}

\Cref{fig:variance-ate-bounds} reports a numerical comparison of all three bounds under two DGPs: (i) a Bregman-Sinkhorn copula DGP with $(F(Y_0), G(Y_1)) \sim c_{\mathsf{S}}^{\rho_0=0.95}$ and pronounced rank violations, and (ii) a Gaussian copula DGP with $(F(Y_0), G(Y_1)) \sim c_{\mathsf{G}}^{\rho_0=0.95}$ with moderate rank violations. Each bound is evaluated across 500 independent draws of $T$, with the potential outcomes held fixed, and compared against the true $\mathop{Var}(\widehat{\tau})$ computed from the unobserved $(Y_0, Y_1)$. The Bregman-Sinkhorn bound $\widehat{\mathop{Var}}_{\mathsf{SB}}$ is substantially tighter than both $\widehat{\mathop{Var}}_{\mathsf{FH}}$ and $\widehat{\mathop{Var}}_{\mathsf{N}}$, which overestimate the true variance across all replications.

The comonotonic coupling ($\rho = 1$) errs in both directions: it underestimates treatment effect heterogeneity by producing an overly concentrated TED, and overestimates the variance of the ATE estimator. The Bregman-Sinkhorn copula corrects both simultaneously: it recovers the full heterogeneity of the TED and yields a tighter variance estimator for the ATE.

\subsection{Covariate Adjustment}
\label{sec:covariate-adjustment}

In observational studies, treatment assignment $T$ may depend on covariates $X$, so $(Y_0, Y_1) \nindep T$. The Bregman-Sinkhorn framework extends naturally to this setting by conditioning on $X$.

\begin{assumption}[Unconfoundedness]
  \label{asmp:no-unobserved-confounding}
  $(Y_0, Y_1) \indep T \mid X = x$ for every $x \in \cX$.
\end{assumption}

\begin{assumption}[Conditional Bregman-Sinkhorn copula]
  \label{asmp:Bregman-Sinkhorn-copula-conditional}
  For every $x \in \cX$, the conditional distribution of $(Y_0, Y_1)$ given $X = x$ is a Bregman-Sinkhorn copula with stickiness parameter $\rho_x \in (0, 1]$:
  \begin{align*}
    (Y_0, Y_1) \mid X = x \;\sim\; \pi_{\rho_x}(\cdot \mid X = x) \in \Pi(\mu_x, \nu_x),
  \end{align*}
  where $\mu_x$ and $\nu_x$ denote the conditional distributions of $Y_0$ and $Y_1$ given $X = x$, respectively.
\end{assumption}
Assumption~\ref{asmp:Bregman-Sinkhorn-copula-conditional} generalizes the conditional rank preservation assumption of \citet{chernozhukov2005iv} and \citet{athey2006identification}, which corresponds to $\rho_x = 1$ for every $x$.

Under Assumption~\ref{asmp:no-unobserved-confounding}, the conditional marginals $\mu_x = \cL(Y_0 \mid X = x) = \cL(Y \mid X = x, T = 0)$ and $\nu_x = \cL(Y_1 \mid X = x) = \cL(Y \mid X = x, T = 1)$ are identified from the observed data. Under Assumption~\ref{asmp:Bregman-Sinkhorn-copula-conditional}, the conditional joint distribution $\pi_{\rho_x}(\cdot \mid X = x)$ is then identified as
$$
\pi_{\rho_x}(\cdot \mid X = x) := \argmax_{\pi \in \Pi(\mu_x, \nu_x)} \bigl\{ \sR(\pi) - \epsilon(\rho_x)\,\mathsf{KL}(\pi \mid \mu_x \otimes \nu_x) \bigr\} \;.
$$
The conditional TED $\cL(Y_1 - Y_0 \mid X = x)$ is recovered by pushing $\pi_{\rho_x}(\cdot \mid X = x)$ forward through $(y_0, y_1) \mapsto y_1 - y_0$, and the marginal TED $\cL(Y_1 - Y_0)$ is obtained by integrating over the distribution of $X$.

When $\rho_x = 1$ for every $x$, the conditional Bregman-Sinkhorn copula reduces to the conditional rank-preserving coupling, yet the marginal distribution of $(Y_0, Y_1)$ need not be rank-preserving---a feature unavailable to any single unconditional Bregman-Sinkhorn copula. To see this, let $X \in \{0, 1\}$ with equal probability, $(Y_0, Y_1) \mid X = 0 \sim \frac{1}{2} \delta_{(0.5,\, 0.25)} + \frac{1}{2} \delta_{(0.75,\, 1)}$, and $(Y_0, Y_1) \mid X = 1 \sim \frac{1}{2} \delta_{(0.25,\, 0.5)} + \frac{1}{2} \delta_{(1,\, 0.75)}$: each conditional is rank-preserving, yet the marginal is not. The conditional Bregman-Sinkhorn mixture therefore generates a strictly richer class of marginal dependence structures, accommodating heterogeneous rank dependence across covariate strata.

We now present an algorithm that computes the conditional rank-preserving Bregman-Sinkhorn copula $\pi_{\rho_x \equiv 1}(\cdot \mid X = x)$ jointly over all $x \in \cX$, for $\cY = \mathbb{R}$. The algorithm generalizes the parabolic Monge--Amp\`ere PDE of \cite{DebLiang2025}, which computes the optimal transport map between $\cL(Y_0 \mid X = x)$ and $\cL(Y_1 \mid X = x)$. The parabolic Monge--Amp\`ere PDE defines a gradient flow that converges to $\dot{\phi}_\infty(y; x)$ as $t \to \infty$: with initial condition $\dot{\phi}_0(y; x) = y$, 
\begin{align*}
\partial_t \dot\phi_t(y; x) = -  \ddot{\phi}_t(y; x) h( \dot{\phi}_t(y; x) ; x) \;,
\end{align*}
where $h(\cdot\,; x)$ is the score of the conditional log-density ratio,
\begin{align*}
h(y ; x) := \partial_y \log \frac{\PP( \dot\phi_t(Y_0; x) = y \mid X = x)}{\PP(Y_1 = y \mid X = x)} \;.
\end{align*}
Setting the right-hand side to zero yields static Monge--Amp\`ere equations \citep{Liang2025DS2,Liang2025DS1} that characterize the optimal transport map $\dot{\phi}_\infty(\cdot\,; x)$ for each $x$:
$$
\dot{\phi}_\infty(\cdot\,; x) \sharp \cL(Y_0 \mid X = x) = \cL(Y_1 \mid X = x) \;.
$$

Under Assumption~\ref{asmp:no-unobserved-confounding}, $h(\cdot\,; x)$ is identified from the observed data $(Y, X, T)$:
\begin{align*}
  h(y ; x) &= \partial_y \log \frac{\PP( \dot\phi_t(Y_0; x) = y \mid X = x, T = 0)}{\PP(Y_1 = y \mid X = x, T = 1)} \quad \text{(unconfoundedness)} \;, \\
&= \partial_y \log \frac{\PP( \dot\phi_t(Y; X) = y \mid X = x, T = 0) \PP(X = x \mid T=0)}{\PP(Y = y \mid X = x, T = 1) \PP(X = x \mid T=1)} \quad \text{(since $\partial_y \log \frac{\PP(X = x \mid T=0)}{\PP(X = x \mid T=1)} = 0$)} \;,\\
& = \partial_y \log \frac{\PP( \dot\phi_t(Y; X) = y,  X = x \mid T = 0)}{\PP(Y = y,  X = x \mid T = 1)} \quad \text{(Bayes' rule)}\;.
\end{align*}
The joint log-likelihood ratio $\log \frac{\PP( \dot\phi_t(Y; X) = y,  X = x \mid T = 0)}{\PP(Y = y,  X = x \mid T = 1)}$ is identified from the observed data and can be estimated by training a binary classifier to predict $T$ from the pooled sample $\{ (\dot\phi_t(Y_i; X_i), X_i) : T_i = 0 \}$ and $\{ (Y_i, X_i) : T_i = 1 \}$. Formal consistency guarantees for the resulting plug-in estimator of the conditional TED are left for future work.

% \section{Discussion}

%%%%%%%%%%%%%%%%%%%%%%%%%%%%%%%%%%%%%%%%%%%%%%
%% Acknowledgements                         %%
%%%%%%%%%%%%%%%%%%%%%%%%%%%%%%%%%%%%%%%%%%%%%%
\section*{Acknowledgements}
TL gratefully acknowledges support from the NSF Career Award (DMS-2042473)
and the Wallman Society of Fellows at the University of Chicago.

%%%%%%%%%%%%%%%%%%%%%%%%%%%%%%%%%%%%%%%%%%%%%%
%% References                               %%
%%%%%%%%%%%%%%%%%%%%%%%%%%%%%%%%%%%%%%%%%%%%%%
\ifbiblatex
%   \printbibliography
\else
  \bibliographystyle{abbrvnat}
  \bibliography{reference}
\fi

%%%%%%%%%%%%%%%%%%%%%%%%%%%%%%%%%%%%%%%%%%%%%%
%% Appendices                               %%
%%%%%%%%%%%%%%%%%%%%%%%%%%%%%%%%%%%%%%%%%%%%%%
\appendix

\section{Proofs}

\subsection{Proofs for Section~\ref{sec:rank-violations}}
\begin{proof}[Proof of Proposition~\ref{prop:maximize-average-rank}]
  Both claims follow from Brenier's Theorem \citep{brenier1991polar} and the identity
  $$
  \sR(\pi) =  - \int_{\cY \times \cY} \frac{1}{2}\| y_0 - y_1 \|^2 \dd \pi(y_0, y_1) + \int_{\cY \times \cY} \frac{1}{2}\| y_0 - y_1 \|^2 \dd \mu(y_0) \dd \nu(y_1) \;.
  $$
  Since the second term is a constant with respect to $\pi$, maximizing $\sR(\pi)$ is equivalent to minimizing the squared-distance transport cost, whose unique solution under absolute continuity is the Brenier map by Brenier's Theorem. For part~(ii), the regularized objective $-\sR(\pi) + \epsilon(\rho) \mathsf{KL}(\pi \mid \mu \otimes \nu)$ is the entropic OT objective up to a constant.
\end{proof}

\begin{proof}[Proof of Proposition~\ref{prop:rank-violations}]
By Theorem~\ref{thm:identifiability}, the unique solution to \cref{eqn:pi-rho} implies that $\pi_\rho$ admits a density with respect to $\mu \otimes \nu$ of the form
\begin{align*}
\PP(Y_1 = \dd y_1 \mid Y_0 = y_0) = \exp\Bigl( \frac{y_0 \cdot y_1 - \phi_{\rho}(y_0) - \psi_{\rho}(y_1)}{\epsilon(\rho)} \Bigr) \dd \nu(y_1) \;.
\end{align*}
Then for any $u \in (0, 1)$ and $\delta < \min(u, 1-u)$
\begin{align*}
&\PP\Bigl( Y_1 \notin [G^{-1}(u-\delta), G^{-1}(u+\delta)] \mid Y_0 = F^{-1}(u) \Bigr) \\
&= \int_{\R \backslash [G^{-1}(u-\delta), G^{-1}(u+\delta)]} \exp\Bigl( \frac{F^{-1}(u) \cdot y_1 - \phi_{\rho}(F^{-1}(u)) - \psi_{\rho}(y_1)}{\epsilon(\rho)} \Bigr) \dd \nu(y_1) \;.
\end{align*}
Since $\nu$ has a strictly positive density and $\cY$ is bounded, by standard Sinkhorn theory for strictly positive continuous costs on compact spaces \citep{csiszar1975divergence,nutz2021introduction}, there exists a bounded continuous pair $(\phi_\rho, \psi_\rho)$ solving the Sinkhorn equations and $\| \phi_\rho(y_0) + \psi_\rho(y_1) \|_\infty < \infty$. Hence $\exp\Bigl( \frac{y_0 \cdot y_1 - \phi_{\rho}(y_0) - \psi_{\rho}(y_1)}{\epsilon(\rho)} \Bigr) \geq c(\cY, \pi_\rho)$ uniformly in $y_1, y_0 \in \cY$ for some constant $c(\cY, \pi_\rho) > 0$. Therefore,
$$
\PP\Bigl( Y_1 \notin [G^{-1}(u-\delta), G^{-1}(u+\delta)] \mid Y_0 = F^{-1}(u) \Bigr) \geq c(\cY, \pi_\rho) \cdot (1-2\delta) \;.
$$
\end{proof}

\subsection{Proofs for Section~\ref{sec:bregman-sinkhorn-copula}}

\begin{proof}[Proof of Proposition~\ref{prop:conditional-moments}]
The conditional distribution of $Y_1 \mid Y_0 = y_0$ is given by \cref{eqn:conditional-distribution}, which is a member of the natural exponential family with natural parameter $y_0 / \epsilon(\rho)$ and log-partition function $\phi_\rho(y_0) / \epsilon(\rho)$. The first and second moments of an exponential family distribution follow from the first and second derivatives of the log-partition function, giving the desired result.
\end{proof}

\begin{proof}[Proof of Proposition~\ref{prop:rank-stickiness}]
The conditional distribution of $G(Y_1) \mid F(Y_0) = u_0$ is given by
\begin{align*}
\PP( G(Y_1) = u_1 \mid F(Y_0) = u_0) &= \frac{\exp\Bigl( \frac{ - B_{\psi_\rho} \bigl( G^{-1}(u_1) \mid \dot\psi^\ast_{\rho} \circ F^{-1}(u_0)  \bigr) }{\epsilon} \Bigr) }{\int_0^1 \exp\Bigl( \frac{ - B_{\psi_\rho} \bigl( G^{-1}(u'_1) \mid \dot\psi^\ast_{\rho} \circ F^{-1}(u_0)  \bigr) }{\epsilon} \Bigr)  \dd u'_1} \;.
\end{align*}
By the strict convexity of $\psi_\rho$, the Bregman divergence $B_{\psi_\rho} \bigl( G^{-1}(u_1) \mid \dot\psi^\ast_{\rho} \circ F^{-1}(u_0)  \bigr)$ has a strict minimum at $u_1 = G \circ \dot\psi^\ast_{\rho} \circ F^{-1} (u_0)$. By the assumption of strictly positive and smooth densities, the map $u_1 \mapsto \psi_\rho \circ G^{-1}(u_1)$ has strictly positive second derivative at $u_1 = G \circ \dot\psi^\ast_{\rho} \circ F^{-1} (u_0)$. Applying the Laplace method, we have for any positive function $h(u_1)$ that is bounded and smooth,
\begin{align*}
\frac{\int_0^1 h(u_1) \exp\Bigl( \frac{ - B_{\psi_\rho} \bigl( G^{-1}(u_1) \mid \dot\psi^\ast_{\rho} \circ F^{-1}(u_0)  \bigr) }{\epsilon_n} \Bigr)  \dd u_1}{\int_0^1 \exp\Bigl( \frac{ - B_{\psi_\rho} \bigl( G^{-1}(u_1) \mid \dot\psi^\ast_{\rho} \circ F^{-1}(u_0)  \bigr) }{\epsilon_n} \Bigr)  \dd u_1} = h\bigl( G \circ \dot\psi^\ast_{\rho} \circ F^{-1} (u_0) \bigr)  + o_n(1)\;.
\end{align*} 
The desired result follows by applying the above display to $h(u_1) = u_1$ and $h(u_1) = u_1^2$, respectively.
\end{proof}

\subsection{Proofs for Section~\ref{sec:identification-inference}}

\begin{proof}[Proof of Theorem~\ref{thm:limit-theorem-rho-1}]
Let $F_n$ and $G_n$ be the empirical CDFs, and define the empirical quantile function
\begin{align*}
G_n^{-1}(u) := \inf\{y \in \cY: G_n(y) \geq u\}, \qquad u \in [0,1].
\end{align*}
By Donsker's theorem and the quantile-process theorem, we have
\begin{align*}
\sqrt{n} \, \bigl( F_n(y) - F(y) \bigr) \rightsquigarrow \bB_0(F(y)), \\
\sqrt{n} \, \bigl( G_n^{-1}(u) - G^{-1}(u) \bigr) \rightsquigarrow \frac{\bB_1(u)}{g(G^{-1}(u))} \;.
\end{align*}
Because the two samples are independent, the convergence holds jointly in
$\ell^\infty(\cY) \times \ell^\infty([0,1])$, and the Gaussian limit has almost surely continuous sample paths, i.e., it lies in $C(\cY) \times C([0,1])$.

Now consider the map
\begin{align*}
\varPhi: \ell^\infty(\cY) \times \ell^\infty([0,1]) \to \ell^\infty(\cY),
\qquad \varPhi( h , k ) := k \circ h.
\end{align*}
Under the stated smoothness and positivity assumptions, $\varPhi$ is Hadamard differentiable at $(F,G^{-1})$, with derivative
\begin{align*}
\dot\varPhi(F,G^{-1})[h,k](y) = \frac{h(y)}{g(G^{-1}(F(y)))} + k(F(y)).
\end{align*}
Applying the functional delta method \citep[Theorem 3.9.4]{van1996weak} to $(F_n, G_n^{-1})$ yields
\begin{align*}
\sqrt{n} \, \bigl( G_n^{-1} \circ F_n(y) - G^{-1} \circ F(y) \bigr) \rightsquigarrow \frac{\bB_0(F(y))}{g(G^{-1} \circ F(y))} + \frac{\bB_1(F(y))}{g(G^{-1} \circ F(y))} \;.
\end{align*}
The limit is mean-zero Gaussian with covariance
\begin{align*}
\Sigma(y, y') = \frac{2\bigl(F(y \wedge y') - F(y) F(y')\bigr)}{g(G^{-1} \circ F(y)) g(G^{-1} \circ F(y'))} \;.
\end{align*}
The limit process has almost surely continuous sample paths on $\cY$. The equivalent representation with $\sqrt{2}\,\bB$ follows since the sum of two independent standard Brownian bridges is distributed as $\sqrt{2}$ times a standard Brownian bridge.
\end{proof}

\begin{proof}[Proof of Theorem~\ref{thm:limit-theorem-rho-less-than-1}]
Let $\epsilon := \epsilon(\rho)$. Define
  \begin{align*}
    h_\rho := \phi_\rho \oplus \psi_\rho,
    \qquad
    h_n := \phi_n \oplus \psi_n,
  \end{align*}
  and the Banach subspace
  \begin{align*}
    \cH_0
    := \Bigl\{ h \in C(\cY\times\cY): h(y_0,y_1)=\phi(y_0)+\psi(y_1),\ \phi,\psi \in C(\cY),\ \int_{\cY} \phi\,\dd\mu = 0 \Bigr\}
  \end{align*}
  with the sup norm. Let $\varPsi: C(\cY\times\cY)\times \sP(\cY)\times \sP(\cY)\to C(\cY\times\cY)$ be
  \begin{align*}
    \varPsi(h ; \mu,\nu)(y_0,y_1)
    := h(y_0,y_1)
    - \epsilon\log \int_{\cY\times\cY}
    \exp\Bigl( \frac{\langle y_0,y_1'\rangle + \langle y_0',y_1\rangle - h(y_0',y_1')}{\epsilon} \Bigr)
    \dd\mu(y_0')\dd\nu(y_1').
  \end{align*}
  For $h=\phi \oplus \psi$ in the separable class, define the gauge-centering projection
  \begin{align*}
    \Pi_0(h) := \bigl(\phi-\int_{\cY} \phi \,\dd\mu\bigr)\oplus\bigl( \psi+\int_{\cY} \phi \,\dd\mu\bigr) \in \cH_0.
  \end{align*}

$h_\rho\in\cH_0$ by the population Sinkhorn normalization. For $h_n$: the fixed-point equation $\varPsi(\cdot;\mu_n,\nu_n)=0$ is invariant under $(\phi_n,\psi_n)\mapsto(\phi_n+c,\psi_n-c)$ for any constant $c$, so we choose the unique $c_n$ with $\int(\phi_n+c_n)\,\dd\mu=0$; the sum $h_n=\phi_n\oplus\psi_n$ is unchanged, giving $h_n\in\cH_0$. Thus, 
\begin{align*}
\varPsi(h_\rho;\mu,\nu)=0,
\qquad
\varPsi(h_n;\mu_n,\nu_n)=0.
\end{align*}

By Lemma \ref{lem:linearization-invertibility-h0},
\begin{align*}
\cA:=\dot \varPsi ( h_\rho;\mu,\nu ) \vert_{\cH_0}:\cH_0\to\cH_0
\end{align*}
is continuously invertible. By Lemma \ref{lem:forcing-clt-h0}, there exists a centered Gaussian element $\bW \in\cH_0$ such that
\begin{align*}
\sqrt{n}\,\Pi_0\bigl(\varPsi(h_\rho;\mu_n,\nu_n)-\varPsi(h_\rho;\mu,\nu)\bigr)
\rightsquigarrow \bW \quad \text{in } \cH_0 \subset C(\cY\times\cY).
\end{align*}

Apply the implicit functional delta method at $(h_\rho;\mu,\nu)$ on the Banach space $\cH_0$:
\begin{align*}
\sqrt{n}(h_n-h_\rho)
= -\cA^{-1}\sqrt{n}\,\Pi_0\bigl(\varPsi(h_\rho;\mu_n,\nu_n)-\varPsi(h_\rho;\mu,\nu)\bigr)+o_p(1)
\end{align*}
in $\cH_0$. The projection $\Pi_0$ appears because $h_n-h_\rho\in\cH_0$ and $\cA:\cH_0\to\cH_0$, so both sides lie in $\cH_0$; applying $\Pi_0$ to the forcing term $-\varPsi(h_\rho;\mu_n,\nu_n)$ (which is separable but not necessarily centered) projects out the component outside $\cH_0$ without affecting the equation. Therefore,
\begin{align*}
\sqrt{n}(h_n-h_\rho) \rightsquigarrow -\cA^{-1}\bW =: \bZ_\rho
\quad \text{in } \cH_0 \cong C(\cY) \oplus C(\cY).
\end{align*}
Since $\cA^{-1}$ is continuous linear and $\bW$ is centered Gaussian, $\bZ_\rho$ is centered Gaussian. Writing $\bZ_\rho = \bZ_{0, \rho} \oplus \bZ_{1, \rho}$ with $\bZ_{0, \rho}, \bZ_{1, \rho} \in C(\cY)$, the convergence holds jointly for the individual potentials in $C(\cY) \times C(\cY)$.
\end{proof}

\subsection{Proofs of Auxiliary Lemmas}

\begin{lemma}[Invertibility on $\cH_0$]
  \label{lem:linearization-invertibility-h0}
  Fix $\rho\in(0,1)$ and write $\epsilon:=\epsilon(\rho)$. Let
  \begin{align*}
    h_\rho=\phi_\rho\oplus\psi_\rho,
    \qquad
    \cH_0:=\Bigl\{h\in C(\cY\times\cY): h(y_0,y_1)= \phi(y_0)+\psi(y_1),\ \phi,\psi\in C(\cY),\ \int \phi \,\dd\mu=0\Bigr\},
  \end{align*}
  and
  \begin{align*}
    \varPsi(h;\mu,\nu)(y_0,y_1)
    := h(y_0,y_1)
    -\epsilon\log\int_{\cY\times\cY}
    \exp\Bigl(\frac{\langle y_0,y_1'\rangle+\langle y_0',y_1\rangle-h(y_0',y_1')}{\epsilon}\Bigr)
    \dd\mu(y_0')\dd\nu(y_1').
  \end{align*}
  Then,
  \begin{enumerate}
    \item[(i)] $\varPsi(\cdot;\mu,\nu)$ is Fr\'echet differentiable at $h_\rho$ on $C(\cY\times\cY)$, denoted by $\dot\varPsi(h_\rho;\mu,\nu)$.
   
    \item[(ii)] The restriction
    \begin{align*}
      \cA:=\dot \varPsi(h_\rho;\mu,\nu)\vert_{\cH_0}:\cH_0\to\cH_0
    \end{align*}
    is a bounded linear isomorphism.
  \end{enumerate}
\end{lemma}
\begin{proof}
Part (i) follows from differentiation under the integral sign and smoothness of $\log$ and the compact domain assumption. Specifically, for any $u \in C(\cY\times\cY)$,
\begin{align*}
  \dot \varPsi(h_\rho;\mu,\nu)[u](y_0,y_1)
  =u(y_0,y_1) + \int_{\cY\times\cY}u(y_0',y_1')\,w_\rho(y_0,y_1;\dd y_0',\dd y_1')
\end{align*}
where the probability transition $w_\rho(y_0, y_1 ; \cdot)$ is given by
$$
  w_\rho(y_0,y_1;\dd y_0',\dd y_1') := \frac{\exp\Bigl(\frac{\langle y_0,y_1'\rangle+\langle y_0',y_1\rangle-h_\rho(y_0',y_1')}{\epsilon}\Bigr)}{\int_{\cY\times\cY}\exp\Bigl(\frac{\langle y_0,y_1''\rangle+\langle y_0'',y_1\rangle-h_\rho(y_0'',y_1'')}{\epsilon}\Bigr)\dd\mu(y_0'')\dd\nu(y_1'')}\,\dd\mu(y_0')\dd\nu(y_1') \;.
$$

For part (ii), $\cA=I+\cK$ where
\begin{align*}
\cK h(y_0,y_1):=\int h(y_0',y_1')\,w_\rho(y_0,y_1;\dd y_0',\dd y_1') \;.
\end{align*}

The operator preserves $\cH_0$ since $h_\rho$ is separable, and therefore the probability transition factorizes $w_\rho(y_0,y_1;\dd y_0',\dd y_1') = w_{\rho}(y_1;\dd y_0') w_{\rho}(y_0;\dd y_1')$,
where 
\begin{align*}
 w_{\rho}( y_1;\dd y_0') &:= \frac{\exp\Bigl(\frac{\langle y'_0,y_1 \rangle - \phi_\rho (y'_0)}{\epsilon}\Bigr)}{ \int_{\cY} \exp\Bigl(\frac{\langle y''_0,y_1 \rangle - \phi_\rho (y''_0)}{\epsilon}\Bigr)\dd \mu(y''_0) } \dd \mu(y'_0) \;, \\
 w_{\rho}( y_0;\dd y_1') &:= \frac{\exp\Bigl(\frac{\langle y_0,y'_1 \rangle - \psi_\rho (y'_1)}{\epsilon}\Bigr)}{ \int_{\cY} \exp\Bigl(\frac{\langle y_0,y''_1 \rangle - \psi_\rho (y''_1)}{\epsilon}\Bigr)\dd \nu(y''_1) } \dd \nu(y'_1) \;.
\end{align*}
For $h = a \oplus b \in \cH_0$,
\begin{align*}
\cK h(y_0,y_1) &:=\int_{\cY \times \cY} h(y_0',y_1')  w_{\rho}( y_1;\dd y_0') w_{\rho}( y_0;\dd y_1') \\
&= \int_{\cY} b(y'_1) w_{\rho}( y_0;\dd y_1')  +   \int_{\cY} a(y'_0) w_{\rho}( y_1;\dd y_0') \\
&= \Pi_0\Bigl( \int_{\cY} b(y'_1) w_{\rho}( y_0;\dd y_1')  +   \int_{\cY} a(y'_0) w_{\rho}( y_1;\dd y_0') \Bigr) \in \cH_0 \;,
\end{align*}
where the last equality holds because $\Pi_0$ does not change the pointwise value of the function: for any $\phi,\psi\in C(\cY)$,
\begin{align*}
  \Pi_0(\phi \oplus \psi)(y_0,y_1)
  :=\Bigl(\phi(y_0)-c\Bigr)+\Bigl(\psi(y_1)+c\Bigr)=\phi(y_0)+\psi(y_1),
  \quad c:=\int\phi\,\dd\mu \;,
\end{align*}
so the second and third lines are equal as functions; $\Pi_0$ merely re-centers the split to place the result in $\cH_0$.
 
We now prove injectivity. Suppose $\cA[u] = 0$ for $u = a \oplus b \in \cH_0$, so $(I+\cK)(a\oplus b)=0$. From the factorization of $\cK$ and the gauge condition $\int a\,\dd\mu=0$, separating the $y_0$- and $y_1$-dependent parts gives
\begin{align*}
  a(y_0) + \int_{\cY} b(y_1')\,w_\rho(y_0;\dd y_1') = 0 \;, \qquad
  b(y_1) + \int_{\cY} a(y_0')\,w_\rho(y_1;\dd y_0') = 0 \;.
\end{align*}
Substituting the second equation into the first yields $a = T_\rho a$, where
\begin{align*}
  T_\rho a(y_0) := \int_{\cY}\int_{\cY} a(y_0')\,w_\rho(y_1;\dd y_0')\,w_\rho(y_0;\dd y_1) \;.
\end{align*}
Since $a \in \cH_0$ satisfies $\int a\,\dd\mu = 0$, any nonzero $a$ must be non-constant. The same argument used in the proof of \cref{thm:identifiability} (applied to the probability kernels $w_\rho(y_0;\cdot)$ and $w_\rho(y_1;\cdot)$, which have strictly positive densities under Assumption~\ref{asmp:no-atoms}) then gives $\|T_\rho a\|_\infty < \|a\|_\infty$, contradicting $a = T_\rho a$. Hence $a = 0$ and $b = 0$. Thus $\ker(\cA) = \{0\}$.

Because $w_\rho(y_0,y_1;\cdot)$ has a jointly continuous density on $\cY \times \cY$, for any bounded set $B\subset\cH_0$ the image $\cK(B)$ is uniformly bounded and, by the joint continuity of $(y_0,y_1)\mapsto\int h\,\dd w_\rho(y_0,y_1;\cdot)$, equicontinuous in $(y_0,y_1)$; Arzel\`a--Ascoli then gives relative compactness of $\cK(B)$ in $C(\cY\times\cY)$, so $\cK:\cH_0\to\cH_0$ is compact. Hence $\cA=I+\cK$ is Fredholm of index zero; injectivity ($\ker\cA=\{0\}$) implies surjectivity by the Fredholm alternative, and bijectivity on the Banach space $\cH_0$ gives a bounded inverse by the bounded inverse theorem.
\end{proof}

\begin{lemma}[CLT on $\cH_0$]
  \label{lem:forcing-clt-h0}
  Under the assumptions of \cref{thm:limit-theorem-rho-less-than-1}, there exists a centered Gaussian process $\bW\in\cH_0$ such that
  \begin{align*}
    \sqrt{n}\,\Pi_0\bigl(\varPsi(h_\rho;\mu_n,\nu_n)-\varPsi(h_\rho;\mu,\nu)\bigr)
    \rightsquigarrow \bW
    \qquad\text{in }\cH_0 \;.
  \end{align*}
\end{lemma}
\begin{proof}
The partial Fréchet derivatives of $\varPsi(h_\rho;\cdot,\cdot)$ at $(\mu,\nu)$ with respect to each marginal are the linear maps $\cB_\mu:\sM(\cY)\to C(\cY)$ and $\cB_\nu:\sM(\cY)\to C(\cY)$ defined by
\begin{align*}
  \cB_\mu[\eta](y_1) := - \epsilon  \int_{\cY} \frac{ \exp\Bigl( \frac{\langle y'_0,y_1 \rangle - \phi_\rho(y_0')}{\epsilon} \Bigr)  }{ \int_{\cY} \exp\Bigl( \frac{\langle y''_0,y_1 \rangle - \phi_\rho(y''_0)}{\epsilon} \Bigr) \dd \mu(y''_0) }  \dd \eta(y'_0) \;. \\
  \cB_\nu[\zeta](y_0) := - \epsilon \int_{\cY} \frac{  \exp\Bigl( \frac{\langle y_0,y'_1 \rangle - \psi_\rho(y'_1)}{\epsilon} \Bigr)  }{ \int_{\cY} \exp\Bigl( \frac{\langle y_0,y''_1 \rangle - \psi_\rho(y''_1)}{\epsilon} \Bigr) \dd \nu(y''_1) } \dd \zeta(y'_1) \;.
\end{align*}

For fixed $h_\rho$, the map $(\eta,\zeta)\mapsto\varPsi(h_\rho;\eta,\zeta)$ is Fr\'echet differentiable at $(\mu,\nu)$ in signed-measure directions, so
\begin{align*}
\varPsi(h_\rho;\mu_n,\nu_n)-\varPsi(h_\rho;\mu,\nu)
=\cB_\mu[\mu_n-\mu]+\cB_\nu[\nu_n-\nu]+r_n \;,
\end{align*}

Consider the function classes indexed by the parameter $y_1\in\cY$ (resp.\ $y_0\in\cY$):
\begin{align*}
  \mathcal{F}_\mu &:= \Bigl\{ f_{y_1}(y_0') := -\epsilon\,\frac{\exp\bigl(\tfrac{\langle y_0',y_1\rangle-\phi_\rho(y_0')}{\epsilon}\bigr)}{\int_\cY \exp\bigl(\tfrac{\langle y_0'',y_1\rangle-\phi_\rho(y_0'')}{\epsilon}\bigr)\dd\mu(y_0'')} \;:\; y_1\in\cY \Bigr\} \;, \\
  \mathcal{F}_\nu &:= \Bigl\{ g_{y_0}(y_1') := -\epsilon\,\frac{\exp\bigl(\tfrac{\langle y_0,y_1'\rangle-\psi_\rho(y_1')}{\epsilon}\bigr)}{\int_\cY \exp\bigl(\tfrac{\langle y_0,y_1''\rangle-\psi_\rho(y_1'')}{\epsilon}\bigr)\dd\nu(y_1'')} \;:\; y_0\in\cY \Bigr\} \;,
\end{align*}
so that $\cB_\mu[\eta](y_1)=\int f_{y_1}\,\dd\eta$ and $\cB_\nu[\zeta](y_0)=\int g_{y_0}\,\dd\zeta$. By compactness of $\cY$, joint continuity of $(y_0',y_1)\mapsto\langle y_0',y_1\rangle$, continuity of $\phi_\rho$ and $\psi_\rho$, and the density positivity assumption (which keeps the denominators bounded away from zero), the classes $\mathcal{F}_\mu$ and $\mathcal{F}_\nu$ are uniformly bounded and equicontinuous as functions of $y_0'$ (resp.\ $y_1'$) uniformly over the index $y_1$ (resp.\ $y_0$). Hence both classes have finite uniform entropy integral and are Donsker. In particular,
\begin{align*}
\cB_{\mu} \bigl[ \sqrt{n}(\mu_n-\mu) \bigr] (y_1) \rightsquigarrow \bG_\mu(y_1) \text{ in }C(\cY) \;,
\qquad
\cB_{\nu} \bigl[ \sqrt{n}(\nu_n-\nu) \bigr] (y_0) \rightsquigarrow \bG_\nu(y_0) \text{ in }C(\cY) \;,
\end{align*}
jointly and independently. Note $\|r_n\|_\infty=o_p(n^{-1/2})$. Therefore,
\begin{align*}
\sqrt{n}\,\bigl(\varPsi(h_\rho;\mu_n,\nu_n)-\varPsi(h_\rho;\mu,\nu)\bigr)
=\cB_\mu\bigl(\sqrt{n}(\mu_n-\mu)\bigr)+\cB_\nu\bigl(\sqrt{n}(\nu_n-\nu)\bigr)+o_p(1) \;.
\end{align*}
Applying the continuous linear maps $\cB_\mu,\cB_\nu$, then $\Pi_0$, gives the claimed convergence with
\begin{align*}
\bW(y_0,y_1):=\Pi_0\bigl( \bG_\mu(y_1)+ \bG_\nu(y_0) \bigr) \in \cH_0 \;.
\end{align*}
\end{proof}

\end{document}